\documentclass[12pt]{article}
\pdfoutput=1
\usepackage[margin=1in]{geometry}
\usepackage{setspace}
\newcommand{\xhdr}[1]{\vspace{2mm} \noindent{\bf #1.}}
\renewcommand{\paragraph}{\xhdr}

\usepackage{amsfonts,amsmath,amssymb,amsthm,graphicx}

\usepackage{color-edits}
\addauthor{yf}{purple}    

\usepackage{changepage}
\usepackage{enumerate}
\usepackage{accents}
\usepackage{float}
\usepackage{subfig}
\usepackage[hidelinks]{hyperref}
\usepackage{caption}

\usepackage{thm-restate}

\usepackage{natbib}
\usepackage[shortlabels]{enumitem}
\usepackage[capitalize]{cleveref}

\usepackage{tikz}
\usepackage{pgfplots}
\usetikzlibrary{patterns}
\usepgfplotslibrary{fillbetween}
\usetikzlibrary{intersections}

\usepackage[ruled,vlined,linesnumbered]{algorithm2e}
\Crefname{algocf}{Algorithm}{Algorithms}

\newtheorem{theorem}{Theorem}
\newtheorem{lemma}{Lemma}
\newtheorem{claim}{Claim}
\newtheorem{proposition}{Proposition}
\newtheorem{corollary}{Corollary}

\newtheorem{definition}{Definition}[section]
\newtheorem{example}[definition]{Example}

\newtheorem{assumption}{Assumption}

\usepackage{ifthen}
\usepackage{fixmath}
\usepackage{xfrac}
\usepackage{sidecap}

\newcommand{\agind}[1][i]{_{#1}}

\newcommand{\constrained}{\hat}

\newcommand{\noaccents}[1]{#1}

\newcommand{\newagentvar}[3][\noaccents]{%
\expandafter\newcommand\expandafter{\csname #2\endcsname}{#1{#3}}%
\expandafter\newcommand\expandafter{\csname #2s\endcsname}{#1{\boldsymbol{#3}}}%
\expandafter\newcommand\expandafter{\csname #2smi\endcsname}[1][i]{#1{\boldsymbol{#3}}_{-##1}}%
\expandafter\newcommand\expandafter{\csname #2i\endcsname}[1][i]{#1{#3}\agind[##1]}%
\expandafter\newcommand\expandafter{\csname #2ith\endcsname}[1][i]{#1{#3}_{(##1)}}%
}

\newagentvar{alloc}{x}

\newagentvar{quant}{q}
\newagentvar[\constrained]{exquant}{\quant}

\newagentvar{typespace}{{\cal T}}

\newagentvar{type}{t}
\newagentvar{val}{v}
\newagentvar{hval}{\bar \val}
\newagentvar{budget}{w}
\newagentvar{wealth}{w}
\newagentvar{cumprice}{\cumpayoff}
\newagentvar{revcurve}{\payoffcurve}
\newagentvar{cumwelfare}{\cumpayoff}
\newagentvar{welcurve}{\payoffcurve}

\newagentvar{price}{p}
\newagentvar{randomprice}{\prices}
\newagentvar{randomquant}{\quants}

\newagentvar{util}{u}

\newagentvar{bid}{b}

\newagentvar{dens}{f}

\newcommand{\exfeasible}{\text{EAF}}

\newagentvar{mech}{{\cal M}}

\newcommand{\Rev}[2][]{\text{\bf Rev}\ifthenelse{\not\equal{}{#1}}{_{#1}}{}\!\left[{\def\givenn{\middle|}#2}\right]}

\newcommand{\Wel}[2][]{\text{\bf Wel}\ifthenelse{\not\equal{}{#1}}{_{#1}}{}\!\left[{\def\givenn{\middle|}#2}\right]}

\newcommand{\lagrange}{\lambda}

\DeclareMathOperator{\OPT}{OPT}

\newcommand{\reals}{{\mathbb R}}

\newcommand{\primed}{^\dagger}
\newcommand{\doubleprimed}{^{\ddagger}}

\newcommand{\R}{\mathbb R}

\newcommand{\MRM}{{\rm MRM}}
\newcommand{\seqquantpayoff}{{\rm SPP}}

\newcommand{\cumpricelinear}{\cumprice^L}
\newcommand{\correspondingquant}{\quant\primed}
\newcommand{\EX}{{\rm EX}}
\newcommand{\EXS}{{\rm EX\primed}}
\newcommand{\EXL}{{\rm EX\doubleprimed}}
\newcommand{\APF}{\tau}
\newcommand{\APFS}{\tau\primed}
\newcommand{\APFL}{\tau\doubleprimed}
\newcommand{\RevPPw}{\Rev[\wealth]{\price}}

\newcommand{\agents}{N}

\newcommand{\exanterelax}{{\rm EAR}}

\newcommand{\myersonReserve}{r^*}

\newagentvar{cumpayoff}{P}
\newagentvar{payoffcurve}{R}

\newcommand{\payoffHull}{\bar{\cumpayoff}}
\newcommand{\concaveHull}{\payoffHull}
\newcommand{\welfareHull}{\payoffHull}

\newcommand{\normalization}{\varpi}
\newcommand{\budgetnum}{K}

\newcommand{\marketclearing}[1][\quant]{\price^{#1}}

\newagentvar{valdist}{F}

\newcommand{\alloclow}{\alloc_L}
\newcommand{\allochigh}{\alloc_H}
\newcommand{\allocq}{\alloc_\quant}

\newcommand{\Payoff}[2][]{\text{\bf Payoff}\ifthenelse{\not\equal{}{#1}}{_{#1}}{}\!\left[{\def\givenn{\middle|}#2}\right]}
\newcommand{\clcumprice}{\hat{\cumprice}_{\lagrange^*}}
\newcommand{\lcumprice}{\cumprice_{\lagrange^*}}

\newcommand{\risk}{\varphi}
\newcommand{\RealizedPayoff}{\Psi}
\newcommand{\distmech}{\pi}
\newcommand{\demand}{d}
\newcommand{\aggregateDemand}{D}
\newcommand{\typedist}{\bar{F}}

\newcommand{\cA}{\mathcal{A}}

\newcommand{\feasibleUF}{\mathcal{P}}

\makeatletter
\DeclareRobustCommand*\cal{\@fontswitch\relax\mathcal}
\makeatother

\DeclareMathOperator*{\argmax}{argmax}

\newcommand{\lbr}[1]{\left\{#1\right\}}
\newcommand{\rbr}[1]{\left(#1\right)}
\newcommand{\cbr}[1]{\left[#1\right]}

\newcommand{\given}{\,\mid\,}

\newcommand{\prob}[2][]{\text{\bf Pr}\ifthenelse{\not\equal{}{#1}}{_{#1}}{}\!\left[{\def\givenn{\middle|}#2}\right]}
\newcommand{\expect}[2][]{\text{\bf E}\ifthenelse{\not\equal{}{#1}}{_{#1}}{}\!\left[{\def\givenn{\middle|}#2}\right]}
\newcommand{\indicator}[2][]{\text{\bf 1}\ifthenelse{\not\equal{}{#1}}{_{#1}}{}\!\left[{\def\givenn{\middle|}#2}\right]}

\newcommand{\dd}{{\mathrm d}}

\newcommand{\abs}[1]{\left| #1 \right |}

\title{Selling to Non-linear Agents:\\ A Demand-Based Auction\thanks{This work is supported by NSF CCF AF \#1618502.
This manuscript combines ``Optimal Auctions vs.\ Anonymous Pricing: Beyond Linear Utility'' as an abstract that appeared at the 20th ACM Conference on Economics and Computation (EC'19),
and ``Simple Mechanisms for Non-linear Agents'' as an extended abstract that appeared at the 34th ACM-SIAM Symposium on Discrete Algorithms (SODA'23).
The authors thank Yi-Chun Chen, Ben Golub, Nima Haghpanah, Deniz Kattwinkel, Ce Li, Harry Pei, Larry Samuelson, Xianwen Shi, Marciano Siniscalchi, Thomas Wiseman and the audiences at Simons Laufer Mathematical Sciences Institute (SLMath), the International Conference on Game Theory at Stony Brook and the World Congress of the Game Theory Society for helpful suggestions.}}

\author{Yiding Feng\thanks{Hong Kong University of Science and Technology. Email: \texttt{ydfeng@ust.hk}}
\quad Jason D. Hartline\thanks{Northwestern University. Email: \texttt{hartline@northwestern.edu}}
\quad Yingkai Li\thanks{National University of Singapore. Email: \texttt{yk.li@nus.edu.sg}}}

\date{}

\begin{document}
\maketitle

\begin{abstract}
Auctions allocate scarce goods, from car licenses to online advertising, to buyers whose utilities are often not linear in money, as with budgets or risk aversion. We ask how a seller should sell to such buyers and what learning their preferences is worth. We propose demand-based auctions, which elicit demand curves and are dominant-strategy incentive compatible for any non-linear utilities. The revenue-maximizing one takes familiar second-price forms and is max-min optimal when the seller observes only demand curves. Under standard assumptions, learning risk attitudes beyond demand curves adds relatively little revenue, as does learning budgets unless they correlate with values.
\end{abstract}


\section{Introduction}

Auctions allocate a wide range of scarce goods, from the right to own a car to shares in newly listed firms and the attention of internet users. Singapore and Shanghai auction the right to register a car, which Beijing instead allocates by lottery \citep{lta2026coe,li2018better}. Initial public offerings in India reserve a portion of the shares for retail investors, who bid for fixed-size lots, and when this portion is oversubscribed, a lottery decides which investors receive one \citep{sebi2020gid}. Online platforms sell advertising space in auctions held each time an ad is displayed. In such markets, a buyer's bid reflects not only how much she values the good, but also how much she can pay and how much risk she is willing to bear. A household's bid for a car license is limited by what it can afford, a retail investor's bid in an Indian offering may not exceed two hundred thousand rupees \citep{sebi2020gid}, and an advertiser bids subject to a budget that platforms spread over many auctions, either by shading her bids or by entering her in only some of them \citep{balseiro2019learning,balseiro2021budget}. Bidders may also be averse to the risk of an uncertain allocation, and empirical work finds substantial risk aversion among auction bidders \citep[e.g.,][]{campo2011semiparametric}.

These markets share a common structure: a seller auctions identical items to buyers who each want one and whose utilities are not linear in money.\footnote{We refer to buyers whose utilities are linear in money as \emph{linear buyers} and to all other buyers as \emph{non-linear buyers}. Following the mechanism design literature, we use the terms buyers and agents interchangeably.} Budgets and risk aversion are two leading sources of such non-linearity: a budget limits what a buyer can pay, and risk aversion leads her to value a chance of winning at less than its expected value. Standard auction theory, in contrast, assumes that a buyer's utility is her value for the item minus her payment, and its central results, the characterization of optimal auctions and their interpretation in terms of marginal revenue \citep{mye-81,BR-89}, rely on this assumption. The difference matters because a seller of indivisible items can sell a buyer a chance of winning rather than the item itself, as lotteries and budget pacing do, and with non-linear utilities, how much a buyer will pay for a given chance is no longer pinned down by her value. Optimal auction theory then offers little guidance, and the optimal mechanism depends on details of the buyers' preferences that a seller rarely knows.

This paper asks how a seller should run an auction for buyers with non-linear utilities, and how much she would gain by learning the details of their preferences, such as their risk attitudes or budgets and how these relate to their values. The question is difficult because non-linear utilities typically make buyers' private information multi-dimensional: a budget-constrained buyer's type consists of her value and her budget, and a risk-averse buyer's type consists of her value and her risk attitude. Optimal mechanisms for multi-dimensional types are notoriously complex and admit closed-form characterizations only in special cases \citep{arm-96,RC-98,RS-03,manelli2007multidimensional,hart2015maximal,daskalakis2017strong}, and our setting is no exception: even for a single buyer with a private budget, the optimal mechanism is characterized only by a differential equation that is generally intractable \citep{CG-00}. At the heart of the difficulty is a ranking problem. With linear utilities, a buyer with a higher value is willing to pay more both for the item and for any chance of winning it, so auction theory can rank buyers by their values. With multi-dimensional types, the ranking depends on what the mechanism offers. A buyer with a high value and a tight budget, for example, is willing to pay less for the item outright than a buyer with a lower value and a slack budget, but more for a small chance of winning it, and differences in risk attitudes raise the same problem. A standard second-price auction sidesteps the problem by ranking bidders by their bids, which are the minimum of their values and budgets, but as we show below, this ranking can be very costly.

To resolve the ranking problem, we propose \emph{demand-based auctions} (Section~\ref{sec:implementation}). Each buyer reports her demand curve, that is, the probability of winning she would purchase at each per-unit price. The auction then ranks buyers by lottery, giving better odds to buyers who demand more at each price, and allocates the items as an auction for linear buyers would, with these random ranks in place of values. In our leading examples, the report is simple: for a budget-constrained buyer, it amounts to a bid and a budget, as in online advertising, and for a risk-averse buyer, to a single bid. Beyond online advertising, auctions that elicit demand rather than a single bid are widely used, for example to sell government debt, electricity, and emission permits \citep{ausubel2014demand}. In the clock auction of C-band spectrum, which drew \$81 billion in winning bids, the Federal Communications Commission asked bidders in each round how many identical blocks of spectrum they demanded at the current prices \citep{fcc2020auction107,fcc2021auction107}. Retail investors in Indian initial public offerings may bid at up to three price levels or submit a ``cut-off'' bid, which commits them to buy at whatever price the offering sets \citep{sebi2020gid}.

The lottery that ranks buyers is designed so that a buyer's odds of being ranked high depend only on her own report and on the seller's forecast of her demand. Each buyer thus faces, in effect, a per-unit price set by the other buyers' reports and receives the probability of winning she asked for at that price. The auction is therefore dominant-strategy incentive compatible for any non-linear utilities, and since it allocates by rank, it is also ex post feasible. Moreover, the seller's revenue depends only on the buyers' expected demand curves (Theorem~\ref{thm:same_payoff} and Proposition~\ref{prop:demand_auction}). In particular, a demand-based auction earns the same revenue from non-linear buyers as the corresponding auction earns from linear buyers with the same expected demand curves. Tools from the linear theory, such as marginal revenue and reserve prices, can therefore be used to design demand-based auctions, whatever the source of the buyers' demands. The seller does not need to know why a buyer demands what she does: a risk-averse buyer and a budget-constrained buyer who purchase the same amount at every price generate the same revenue.

Applying marginal revenue maximization analysis to demand curves yields the revenue-maximizing demand-based auction, which extends the optimal auction of \citet{mye-81} and \citet{BR-89} to non-linear buyers. With a single item and i.i.d.\ buyers, and under a regularity condition, it takes familiar forms in our leading examples, which reflect how each form of non-linearity affects the ranking (Section~\ref{sub:illustrations} and Figure~\ref{fig:quantile_dist}). A risk-averse buyer's rank depends only on her value, whatever her risk attitude, so with risk-averse buyers, the auction is the second-price auction with a monopoly reserve, the optimal auction for risk-neutral buyers. The seller thus need neither elicit nor know risk attitudes. A budget-constrained buyer's rank, in contrast, is random when her budget binds, and with budget-constrained buyers, the auction is a second-price auction with a monopoly reserve run on \emph{randomized budget caps}: each bid is the bidder's value capped by her budget scaled up by an independent random factor, $b_i=\min\{\val_i,\sfrac{\wealth_i}{U_i}\}$ with $U_i$ uniform on $[0,1]$. A bidder whose budget covers her price wins and pays that price, as in the standard auction, whereas a bidder whose budget falls short pays her budget for a proportional chance of winning. The standard second-price auction excludes such bidders altogether and can lose almost all of the revenue of the demand-based auction (Example~\ref{ex:spa_budget}).

These formats resemble practice, where randomized allocation is common. Beijing allocates license plates by lottery \citep{li2018better}, and Indian initial public offerings combine both features of our auction: retail investors report demand at several prices, and a lottery allocates lots when demand exceeds supply \citep{sebi2020gid}. Advertising platforms pace budget-constrained advertisers, for example by entering each advertiser in a randomly selected fraction of auctions that her budget can afford \citep{balseiro2021budget,conitzer2022pacing}, and randomized budget caps do the same within a single auction. Moreover, such markets need little information: buyers report only their demand curves, and the seller needs only the expected demand curves, which can be estimated from purchases at posted prices \citep[cf.][]{bajari2015machine}.

A demand-based auction ignores any information the seller may have about why buyers demand what they do. How much this costs depends on what she knows, and Section~\ref{sec:guarantees} considers two extremes. If the seller observes only the expected demand curves, the revenue-maximizing demand-based auction is max-min optimal (Theorem~\ref{thm:robust_optimality}): no mechanism guarantees more revenue across all specifications of preferences consistent with these curves. A seller who is uncertain about the source of non-linearity should therefore not try to screen buyers with menus of lotteries, although, as the second-price example above shows, she must still adapt the auction to non-linearity. If instead the seller knows the buyers' utility functions and type distributions, a demand-based auction is no longer optimal, but for any non-linear utilities, it loses at most a factor $\gamma\zeta$ relative to the optimal mechanism, which may offer arbitrary menus of lotteries (Theorem~\ref{thm:meta thm expected utility}).\footnote{Unlike for linear buyers \citep{GGKMS-13}, restricting attention to dominant-strategy incentive compatible mechanisms is a genuine restriction for non-linear buyers \citep{FH-18,FLLT-18}. Our benchmark is the optimal Bayesian incentive compatible mechanism.} The factor $\gamma$ is the auction's approximation factor for linear buyers and is inherited from the linear theory. For the revenue-maximizing auction with $m$ units, $\gamma\leq \sfrac{1}{(1-\frac{1}{\sqrt{2\pi m}})}$ \citep{yan-11}, which approaches one as the market grows (Corollary~\ref{cor:mrm}). The factor $\zeta$ reflects non-linearity alone and is defined for any utility model by the problem of selling to a single buyer: it measures how much more revenue a seller could extract from one such buyer with an optimal menu of lotteries than with posted per-unit prices. Since the revenue of a demand-based auction does not depend on what the seller knows, the two results bracket all intermediate cases, and $\zeta$ bounds, up to the factor $\gamma$, what we call the \emph{value of detail}: the seller's gain from learning the details of the buyers' preferences beyond the demand curves she already observes.

Under standard assumptions on value distributions, the value of detail in our leading examples is bounded by a constant, except when budgets are correlated with values (Section~\ref{sec:revenue close}). For risk aversion of any degree, $\zeta\leq e$. For budgets drawn independently of values, $\zeta\leq 3$, and $\zeta=1$ when budgets are public. When budgets are correlated with values, in contrast, the seller can use budgets to screen buyers and extract nearly their full values, and $\zeta$ can be unbounded (Example~\ref{example:correlated_main}). These bounds are worst-case guarantees, and their comparison across classes of preferences is more informative than their exact values \citep[see][for further discussion]{har-12,hartline2015non}. The comparison has two implications. First, learning buyers' risk attitudes is worth relatively little. Risk aversion does not change how a buyer responds to per-unit prices, so it does not affect the revenue of a demand-based auction. It affects only the revenue that a fully informed seller could obtain by insuring the buyer, and this gain is bounded regardless of the degree of risk aversion. Second, what makes budgets worth learning is not their presence but their correlation with values. The same approach applies to other forms of non-linearity, since bounding $\zeta$ requires analyzing only a single buyer, and it extends to objectives other than revenue, such as welfare, and to markets in which buyers have different forms of non-linearity (Sections~\ref{sub:general_objectives} and~\ref{sub:heterogeneous}).

Our revenue guarantees also extend to simpler demand-based auctions, such as sequential posted pricing, in which the seller offers per-unit prices to buyers one at a time until the items run out and buyers report nothing (Section~\ref{sub:spp}). Under every utility model, sequential posted pricing earns at least a $1-\frac{1}{\sqrt{2\pi m}}$ fraction of the revenue of the revenue-maximizing demand-based auction, a fraction that approaches one as the market grows. It is therefore max-min optimal up to this factor and satisfies the same $\gamma\zeta$ bound as the revenue-maximizing auction (Corollary~\ref{cor:spp}). In numerical examples with 200 buyers and 20 units, the loss is far smaller than these worst-case bounds: sequential posted pricing is within a factor of $1.153$ of an upper bound on the optimal revenue for risk-averse buyers and $1.095$ for budget-constrained buyers (Section~\ref{sec:example}).

\subsection{Related Work}
Mechanism design for non-linear agents is widely studied in the literature.
Most relevantly, \citet{Bai-17} considers mechanism design without linear utilities and proposes a novel random allocation mechanism based on the reported demand curves. However, the mechanism implemented in \citet{Bai-17} offers a uniform price to all agents, which may be far from optimal because it does not price discriminate among agents based on their prior heterogeneity \citep[cf.][]{jin2021tight}. Our demand-based auctions also elicit demand curves, but the per-unit price each agent faces depends on her own expected demand curve, which allows the seller to price discriminate based on prior heterogeneity.
\citet{kazumura2020mechanism} characterize the set of dominant-strategy incentive compatible mechanisms in environments without linear utilities and provide a general procedure for optimizing over this class. However, depending on the utility models, obtaining a closed-form characterization using their procedure remains challenging. The main contribution of our paper is a construction that ranks non-linear agents by lotteries based on their reported demands, which yields demand-based auctions, a class of dominant-strategy incentive compatible mechanisms for non-linear agents. We show that the revenue-maximizing demand-based auction is max-min optimal when the seller observes only demand curves and that, for broad classes of non-linear utilities, demand-based auctions are approximately optimal even compared with Bayesian incentive compatible mechanisms when the seller knows the utility models.\footnote{\citet{kushnir2019equivalence} establish the equivalence between dominant-strategy incentive compatible and Bayesian incentive compatible mechanisms under the assumptions of increasing differences over distributions and a convex-valued property. These assumptions are not made in our paper.}

\citet{AFHH-13} also reduce approximation for non-linear agents to approximation for linear agents.
Their framework converts marginal revenue mechanisms for linear agents into mechanisms for non-linear agents. Unlike our approach, which uses single-agent price-posting mechanisms as a building block, their framework involves single-agent ex ante optimal mechanisms, which in general are hard to characterize (e.g., for risk-averse utilities).
Due to this difference, the framework and its implementation in \citet{AFHH-13} are more complex, and the resulting mechanisms are in general not dominant-strategy incentive compatible for non-linear agents.

Our leading examples relate to a large literature on auctions with risk-averse or budget-constrained bidders.
Most results for agents with risk-averse utilities consider the
comparative performance of the first- and second-price auctions
\citep[cf.][]{Hol-80,CG-06}.
\citet{mat-83} and \citet{MR-84},
however, characterize the optimal mechanisms for symmetric agents with
constant absolute risk aversion and with more general risk-averse utilities, respectively.

\citet{che1998standard} show that revenue equivalence among standard auctions fails when bidders are financially constrained. For symmetric agents with \emph{public}
budgets in single-item environments, \citet{LR-96} and \citet{mas-00} respectively study the revenue-maximization and
welfare-maximization problems.
Their results are generalized to i.i.d.\ values but asymmetric public budgets in \citet{BS-18}
and to symmetric agents with
uniformly distributed private budgets in \citet{PV-14}.
\citet{CG-00} consider the single-agent problem with a \emph{private}
budget and a value distribution with declining marginal
revenue, and characterize the optimal mechanism by a differential
equation.
In environments without unit-demand constraints, a price-posting mechanism is optimal for selling a divisible good to a
continuum of agents with private budgets if their valuations are
regular with decreasing density \citep{Richter-18}
and is a $2$-approximation in general \citep{abr-06}.
For more general settings, no closed-form characterizations are known, and algorithmic solutions are provided in \citet{AFHHM-12} and \citet{DW-17}.

The robust optimality of simple mechanisms in complex environments has also been established in \citet{carroll2015robustness,Car-17}, \citet{brooks2021optimal}, \citet{hartline2020benchmark}, and \citet{guo2023regret}.
Our robustness result shows that, for auctions with non-linear utilities, the revenue-maximizing demand-based auction is robustly optimal, since the revenue of demand-based auctions does not depend on the details of the utility models.
In addition, our approximation results complement that line of work by showing that the seller may still wish to adopt simple mechanisms when she knows these details, since their performance is close to optimal.

It is well known that simple mechanisms are approximately optimal for linear agents.
\citet{CHMS-10} show that the approximation ratio for order oblivious posted pricing is $2$, and
\citet{yan-11} shows that the approximation ratio for sequential posted
pricing is at most $\sfrac{e}{(e-1)}$ and approaches one in multi-unit markets as the number of units grows.
For non-linear agents in
matroid environments, \citet{CMM-11} show that a simple lottery
mechanism is a constant approximation to the optimal ex post
individually rational mechanism for agents with monotone-hazard-rate
valuations and private budgets. In contrast, our approximation
results are relative to the optimal interim individually rational mechanism, whose revenue can be arbitrarily larger than the
benchmark of \citet{CMM-11}.

\section{Model}
\label{sec:prelim}
We consider the problem of selling $m$ units of indivisible and identical items
to~$n$ non-linear agents where each agent has unit demand for the items.
In this paper, since all agents have unit demand, it is without loss to consider allocations to each agent as a lottery over $\{0,1\}$.

\paragraph{Agent's Utility Models}
There is a set of agents $\agents$ where $\abs{\agents} = n$.
An agent's \emph{utility model} is defined as
$(\typespace, \typedist, \util)$ where $\typespace, \typedist$, and
$\util$ are the type space, type distribution and utility function. The outcome
for an agent is the distribution over the pair $(\alloc,
\price)$, where allocation $\alloc \in \{0,1\}$ and payment $\price
\in \R_+$.
The utility function $\util$ of the agent is a mapping
from her private type and the realized outcome to her
von Neumann-Morgenstern expected
utility for the outcome.
Types are private and independently distributed across agents.
Throughout this paper, we use the subscript $i$ to denote agent $i$ in multi-agent settings,
and we will drop the subscripts without ambiguity when we discuss single-agent problems. 

There are several specific utility models of interest in this paper.
\begin{itemize}
\item\textbf{Linear utility:}
The agent's private type is her value $\val$ of the good.
Given allocation~$\alloc$ and payment $\price$,
her utility is
\begin{align*}
\util(\val,\alloc,\price) = \val\cdot\alloc - \price.
\end{align*}

\item \textbf{Non-linear utility:}
For a non-linear agent with type $\type$, her utility for allocation~$\alloc$ and payment~$\price$
is $\util(\type,\alloc,\price)$.
In our general reduction framework for non-linear utilities, we do not impose any specific structural properties on the utility function $\util$.
However, there are several important special cases of non-linear utilities to which we apply our general framework.

\begin{itemize}
\item
\textbf{Risk-averse utility:}
The agent's type is a pair $\type=(\val, \risk)$ where
$\val$ is the private value and $\risk$ is a strictly increasing and concave function that represents the agent's private risk attitude.
We assume that $\risk(0) = 0$.
Given allocation~$\alloc$ and payment $\price$,
her utility is
\begin{align*}
\util((\val,\risk),\alloc,\price) = \risk(\val \alloc - \price).
\end{align*}

\item
\textbf{Private-budget utility:}
The agent's type is a pair $\type=(\val, \budget)$ where
$\val$ is the private value and $\budget$ is the private budget constraint.
Given allocation~$\alloc$ and payment $\price$,
her utility is
\begin{align*}
\util((\val,\budget),\alloc,\price) =
\begin{cases}
\val \alloc - \price & \price\leq \budget,\\
-\infty & \price> \budget.
\end{cases}
\end{align*}

\end{itemize}

\end{itemize}
We use $\valdist$ to denote the marginal distribution over values for agents with linear utilities, risk-averse utilities, or private-budget utilities.

\paragraph{Mechanisms}
By the revelation principle \citep{mye-81}, we focus on revelation mechanisms $\{(\alloci,\pricei)\}_{i\in \agents}$,
where agents simultaneously submit bids $\{\bid_i\}_{i\in \agents}$ from their type spaces
to the mechanism,
and each agent $i$ gets allocation
$\alloc_i(\{\bid_i\}_{i\in \agents})$ with
payment $\price_i(\{\bid_i\}_{i\in \agents})$.
An allocation is feasible if $\sum_i \alloci \leq m$ where $m$ is the number of items.

We consider mechanisms that satisfy
\emph{Bayesian incentive compatibility} (BIC),
i.e., no agent can gain strictly higher
expected utility than reporting her private
type truthfully
if
all other agents are reporting their private types truthfully,
and \emph{interim individual rationality} (IIR),
i.e., the expected utility
is non-negative for all agents and all private types
if all agents are reporting their private types
truthfully.
For later discussion, we also call a mechanism
\emph{dominant-strategy
incentive compatible} (DSIC) if
no agent can gain strictly higher expected utility
than reporting her private type truthfully, regardless of
other agents' reports.
Note that DSIC coincides with BIC in single-agent environments.

\paragraph{Pricing and Demand}
We study mechanisms based on simple
per-unit posted pricing.
The definition of the per-unit pricing mechanisms for general multi-agent environments will be introduced formally in \cref{sec:implementation}.
In this section, we will focus on pricing for single-agent environments.
\begin{definition}\label{def:per-unit}
A \emph{per-unit price} $\price$ denotes the menu $\{(z, z\cdot\price): z \in [0,1]\}$.
The agent who chooses an option $(z, z\cdot\price)$ receives an item with probability~$z$,
and pays $z\cdot\price$ regardless of the allocation.\footnote{
We focus on the all-pay format in our definition of per-unit posted pricing mechanisms.
The payment format does not matter for linear agents,
but it affects the incentives for non-linear agents.
Our general reduction framework can also be extended to other payment formats such as winner-pays, with the demand curves, and hence the revenue, defined for the chosen format.
}
\end{definition}
Given per-unit price $\price$,
the demand of the agent with type $\type$ is
\begin{align*}
\demand^{\type}(\price) = \argmax_{z\in[0,1]}\,\, z\cdot \util(\type,1,z\cdot\price)
+ (1-z) \cdot \util(\type,0,z\cdot\price).
\end{align*}
For example, a budget-constrained agent with value $\val$
and budget $\budget$
given \emph{per-unit price} $\price$
will purchase the lottery
$\demand = \min\{1, \sfrac{\budget}{\price}\}$
if $\val \geq \price$,
and purchase the lottery $\demand = 0$ otherwise.
Throughout the paper, we will assume that given any price $\price\geq 0$, the utility of the agent is upper semi-continuous in chosen lottery~$z$.
This implies that the demand maximizing the utility can always be attained.
Additionally, we assume that for any type $\type\in\typespace$, $1 \in \demand^{\type}(0)$
and there exists $\price > 0$ such that $\demand^{\type}(\price) = 0$.
We denote the expected demand correspondence as $\aggregateDemand(\price) = \expect[t\sim\typedist]{\demand^{\type}(\price)}$.
Note that the expectation is taken over each agent's own type.

\paragraph{Ordinary Goods Assumption}
In this paper, our focus is on the sale of ordinary goods, where the demand of the agent is non-increasing in the offered per-unit price. This assumption excludes Giffen goods or Veblen goods and is satisfied by many common utility models, such as risk-averse utility or private-budget utility.

\begin{assumption}\label{asp:ordinary good}
The item is an \emph{ordinary good},
i.e.,
for any type $\type$, the demand $\demand^{\type}(\price)$
is upper hemi-continuous and weakly decreasing in $\price$ in strong set order.
\end{assumption}
In the special case where the optimal demand is unique, the ordinary goods assumption is equivalent to the assumption that the demand curve is weakly decreasing in per-unit price~$\price$.

\paragraph{Seller's Objective}
We consider general payoff maximization for the seller, where welfare maximization,
revenue maximization and their convex combinations
are special cases of payoff maximization.

Specifically, we assume that the seller is an expected utility maximizer.
The seller's payoff function is a mapping from the type-outcome pair of each agent to a real value, and it is additive and separable across different agents.
That is, there exist payoff functions $\{\RealizedPayoff_i\}_{i\in \agents}$
such that given allocations $\{\alloci\}_{i\in \agents}$ and payments $\{\pricei\}_{i\in \agents}$
for type profile $\{\typei\}_{i\in \agents}$,
the payoff of the seller is
\begin{align*}
\RealizedPayoff(\{\typei\}_{i\in \agents}, \{\alloci\}_{i\in \agents}, \{\pricei\}_{i\in \agents})
= \sum\nolimits_{i\in \agents} \RealizedPayoff_i(\typei, \alloci, \pricei).
\end{align*}
A special case of the payoff function is the revenue,
where $\RealizedPayoff_i(\typei, \alloci, \pricei) = \pricei$ for any agent $i$.
Throughout the main text, we focus on revenue maximization, which allows us to present the construction and its economic implications most transparently. Our results extend to general payoff functions, including welfare maximization and convex combinations of welfare and revenue. We discuss these extensions in Section~\ref{sub:general_objectives} and provide the details in Appendix~\ref{apx:general objectives}.

\paragraph{Price-Posting Revenue Curves}
As in \citet{BR-89}, our mechanisms are built from single-agent price posting. We first define the market clearing price for selling the item to an agent with an ex ante supply constraint $\quant\in[0,1]$, i.e., the probability with which the item is sold to her in expectation.
\begin{definition}
\label{def:mcp}
For any $\quant\in[0,1]$, the \emph{market clearing price $\marketclearing$} for the ex ante supply constraint $\quant$ is the highest per-unit price at which there exists a tie-breaking rule for demands such that the expected demand is $\quant$, i.e., $\aggregateDemand(\marketclearing) = \quant$.\footnote{In the case where the demand of the agent is not unique, we may need to consider randomized tie-breaking rules for the demand correspondence in order to satisfy the expected demand constraint.}
\end{definition}
The existence of the market clearing price is implied by the upper hemi-continuity of the demand correspondence, and the ordinary goods assumption implies that $\marketclearing$ is weakly decreasing in $\quant$.
The price-posting revenue curve describes the expected revenue in the single-agent problem from posting per-unit prices, as a function of the ex ante supply constraint.
\begin{definition}
\label{def:ppcurve}
The \emph{price-posting revenue curve} $\cumpayoff(\quant)=\quant\cdot\marketclearing$, for $\quant\in[0,1]$, is the expected revenue from posting the market clearing price $\marketclearing$.
\end{definition}
For general payoff functions, the \emph{price-posting payoff curve} is defined analogously (see Appendix~\ref{apx:general objectives}). If the seller posts per-unit prices to a single agent, her revenue is determined by the price-posting revenue curve, regardless of the details of the agent's utility model.
\cref{sec:implementation} extends this observation to multiple agents.

\section{Demand-Based Auctions}
\label{sec:implementation}
For linear agents, every deterministic, dominant-strategy incentive compatible auction, including the revenue-optimal one, can be described in a simple way \citep{BR-89,AFHH-13}. Each agent is offered the item at a price that depends only on the other agents' bids, set by comparing agents through the quantiles of their values. Since a linear agent buys either the whole item or nothing, this price is effectively a per-unit price, and we call such mechanisms \emph{per-unit pricing mechanisms}. This section extends them to non-linear agents with two changes. First, each agent can buy any probability of winning at her per-unit price. Second, since her type may be multi-dimensional, her quantile is drawn from a distribution determined by her demand curve rather than computed from her value.
Because these quantiles depend only on demand curves, the seller can run the resulting mechanisms by asking agents to report their demand curves, without knowing their utility functions. We call this implementation a \emph{demand-based auction} (\cref{sub:demand_reports}). The two terms thus refer to the same mechanisms: a per-unit pricing mechanism describes the menu each agent faces, which is what our analysis uses, and a demand-based auction describes how the seller runs it. Demand-based auctions are dominant-strategy incentive compatible and ex post feasible for arbitrary non-linear agents, and they generate the same revenue as their counterparts for linear agents with the same expected demand curves. Consequently, as for linear agents \citep{BR-89}, marginal revenue maximization yields the highest revenue among them. How they compare with all mechanisms, including those that tailor lotteries to the details of the agents' utility functions, is the subject of Section~\ref{sec:guarantees}.

\subsection{Per-Unit Pricing Mechanisms for Linear Agents}
\label{sub:quantile_pricing}
We begin by formalizing the per-unit pricing mechanisms introduced above. A \emph{per-unit pricing mechanism} offers each agent a per-unit price that depends only on the other agents' reports, and the agent chooses how much of the item to buy at that price. The class includes auctions, not only posted prices: in a second-price auction, for instance, each agent's price is the highest competing bid.
This subsection recaps these mechanisms for linear agents \citep{mye-81,BR-89,AFHH-13}, and \cref{sub:implementation} extends them to non-linear agents.
We first recall the concept of \emph{quantiles}, which rank linear agents by their values, and then describe per-unit pricing mechanisms in quantile space.

\paragraph{Quantiles for Linear Agents} For an agent with linear utility,
the quantile of value~$\val$ given value distribution $\valdist$ is defined as
\begin{align*}
\quant(\val) \in \cbr{\Pr_{z\sim\valdist}[z>\val], \Pr_{z\sim\valdist}[z\geq \val]},
\end{align*}
where quantile $\quant(\val)$ is drawn uniformly randomly if the above interval is not degenerate.
For distributions without point masses, $\quant(\val) = 1-\valdist(\val)$.
Conversely, the value corresponding to quantile $\hat{\quant}$ is
$\val(\hat{\quant}) \triangleq \sup\{z:\quant(z) \geq \hat{\quant}\}$.
Note that for any distribution $\valdist$ over value $\val$,
quantile~$\quant$ is uniformly distributed on $[0,1]$.

Based on this mapping for linear agents, any mechanism in the value space can be equivalently converted to a mechanism in the quantile space, and vice versa.
Moreover, for a linear agent, the market clearing price $\marketclearing$ coincides with the value corresponding to quantile $\quant$, i.e., $\marketclearing=\val(\quant)$.
Without ambiguity, we also refer to $\cumpayoff(\quant)$ as the price-posting revenue for quantile $\quant$.

\paragraph{Mechanisms in Quantile Space}
For linear agents, every deterministic,
dominant-strategy incentive compatible (DSIC)
and interim individually rational (IIR) mechanism (e.g., the Bayesian optimal mechanism)
can be implemented as a per-unit pricing mechanism \citep{AFHH-13}.\footnote{Note that for linear agents,
any BIC and IIR mechanism can be converted to a DSIC and IIR mechanism \citep{GGKMS-13},
and any randomized mechanism can be converted to a deterministic mechanism through purification \citep{chen2019equivalence} when type distributions are continuous.
Moreover, the revenue-optimal mechanism for linear agents is deterministic and DSIC.}
Based on our definition of quantiles for linear agents,
the per-unit pricing mechanism for linear agents in quantile space
is defined as follows:
\begin{itemize}
\item[-] each agent $i$'s reported value $\vali$ is mapped to quantile $\quanti$ based on~$\valdist_i$;
\item[-] there is a threshold function $Q_i(\quant_{-i})$ for each agent $i$ such that agent $i$ is offered a per-unit price $\price_{i}^{Q_i(\quant_{-i})}$
where $\quant_{-i}\triangleq\{\quant_j\}_{j\in N\backslash\{i\}}$.
Agent $i$ receives an item if and only if $\quanti\leq Q_i(\quant_{-i})$.\footnote{It is possible that in per-unit pricing mechanisms for linear agents,
based on the tie-breaking rules,
agent~$i$ receives an item only if $\quanti < Q_i(\quant_{-i})$ for some $\quant_{-i}$.
In this case, by slightly abusing the notation,
we allow $Q_i(\quant_{-i})$ to take values such as $z^+$ for $z\in[0,1]$,
and we say $\quanti\leq z^+$ if and only if $\quanti < z$.
}
\end{itemize}
Essentially, $Q_i(\quant_{-i})$ is the threshold on allocation in the quantile space, and the per-unit pricing mechanism is uniquely determined by this profile of threshold functions.

Since the distribution over quantiles is uniform on $[0,1]$ regardless of the value distribution,
the expected revenue of the per-unit pricing mechanism~$\mech$ with threshold functions $\{Q_i\}_{i\in \agents}$ is uniquely pinned down by the price-posting revenue curves $\{\cumpayoff_i\}_{i\in \agents}$ as
\begin{align*}
\mech(\{\cumpayoff_i\}_{i\in \agents}) =
\expect[\{\quant_i\}_{i\in \agents} \sim {U[0,1]^n}]{\sum_{i\in \agents} \cumpayoff_i(Q_i(\quant_{-i}))}.
\end{align*}
In the next subsection, we will use this formula to derive the performance of per-unit pricing mechanisms for non-linear agents.

\subsection{Per-Unit Pricing Mechanisms for Non-linear Agents}
\label{sub:implementation}
In this subsection, we provide a quantile-space implementation of per-unit pricing mechanisms for non-linear agents.
For this purpose, we need to first provide a formal definition of quantiles for non-linear agents.
For linear agents, the definition of quantiles is intuitive since higher values are mapped to lower quantiles according to the cumulative distribution function of the value distribution. This ranking is the same for every menu a mechanism may offer, since an agent with a higher value demands weakly more at every per-unit price.
In contrast, the construction of quantiles for non-linear agents is not obvious since the private types of non-linear agents can be multi-dimensional
and there may not exist a natural order on types that is consistent for all mechanisms.
For example, it is unclear how to rank a more risk-averse agent with high value against a less risk-averse agent with low value,
or an agent with high value and low budget against another agent with low value and high budget. In both cases, the comparison depends on the menu. The risk-averse agent with a high value is willing to pay more than the other agent for the item outright, but may be willing to pay less for a lottery over it, and the reverse holds for the agent with a high value and a low budget.
We therefore introduce a definition of quantiles for non-linear agents based on their demands given market clearing prices under different ex ante supply constraints.

\paragraph{Quantiles for Non-linear Agents}
For a non-linear agent with type $\type$,
define function
$H^{\type}(\quant) = \demand^{\type}(\marketclearing)$
where $\demand^{\type}(\marketclearing)$ is the demand of type $\type$ under market clearing price $\marketclearing$ that sells the item to the agent with probability $\quant$.
Under the ordinary goods assumption, we can focus on tie-breaking rules such that the demand $\demand^{\type}(\marketclearing)$ is weakly increasing in $\quant$,
which implies that function $H^{\type}(\quant)$
is weakly increasing in $\quant$ for all types~$\type$.\footnote{In the case where the demands of the agent need to be randomized to satisfy the expected demand constraint, we consider a randomization over tie-breaking rules such that
the demand curve $\demand^{\type}(\marketclearing)$ is weakly increasing for all realized selections.
Moreover, for different selections of demand curve $\demand^{\type}(\marketclearing)$,
the defined function $H^{\type}(\quant)$ may be different for any given $\type$ and $\quant$.
However, the results in our paper hold for all selection rules satisfying the required weak monotonicity. }
Moreover, $H^{\type}(0)=0$ and $H^{\type}(1)=1$,
and thus $H^t$ can be viewed as a distribution supported on $[0,1]$ with $H^{\type}(\quant)$ being the cumulative distribution function of quantile for any type $\type$.
\begin{definition}[Quantiles]
\label{def:quantile_nonlinear}
For any non-linear agent with type $\type$ and demand curve $\demand^{\type}(\cdot)$,
the randomized quantile $\quant$ for type $\type$ is drawn from the distribution $H^t$
with $H^{\type}(\quant) \triangleq \demand^{\type}(\marketclearing)$.
\end{definition}

Our definition of quantiles for non-linear agents
maps the potentially multi-dimensional types of non-linear agents into single-dimensional quantiles based on the demand curves for per-unit pricing.
Moreover, given any market clearing price, types with higher demands are mapped to lower quantiles with higher probabilities,
and hence are more likely to win an item in per-unit pricing mechanisms
where items are allocated based on the thresholds in the quantile space.
Note that the definition of quantile for non-linear agents encompasses the definition for linear agents. In particular, for a linear agent with continuous value distribution and realized value $\val(\quant)$, the distribution $H^{\val(\quant)}$ is a point-mass distribution at $\quant$
and $H^{\val(\quant)}(\cdot)$ is a step function at $\quant$.
The distribution $H^{\type}$ also takes a simple form in the two non-linear utility models we study, and the contrast between the two is useful for understanding our results. Figure~\ref{fig:quantile_dist} illustrates both cases.
\begin{itemize}
    \item For a \emph{risk-averse} agent, a per-unit price never induces an interior lottery (Claim~\ref{clm:risk_demand}): by the concavity of $\risk$, she prefers either to buy outright or not to buy at all. Her demand is therefore $\demand^{\type}(\price) = \indicator{\val\geq\price}$, which does not depend on $\risk$, so $H^{\type}$ is a step function at the quantile of her value, exactly as for a linear agent with the same value. In other words, risk aversion does not affect her behavior under per-unit pricing, and per-unit pricing mechanisms treat her as a linear agent with the same value (Section~\ref{sub:illustrations}).
    \item For a \emph{budget-constrained} agent facing per-unit price $\price$, the demand is $\min\{1,\wealth/\price\}$ whenever $\val\geq\price$, which is an interior lottery whenever her budget binds. Her quantile distribution $H^{\type}$ is non-degenerate: it is zero until the market clearing price falls to her value, jumps to her budget share at that point, and increases to one as the price decreases further and her budget becomes slack. Hence, two agents with the same value but different budgets are ranked differently, and the mechanism randomizes to accommodate them.
\end{itemize}

\begin{figure}[t]
\centering
\subfloat[risk-averse agents]{
\begin{tikzpicture}
\begin{axis}[width=0.46\textwidth, height=0.36\textwidth,
  axis lines=left, xmin=0, xmax=1.02, ymin=0, ymax=1.08,
  xtick={0,0.1,0.4,1}, ytick={0,0.5,1},
  xlabel={quantile $\quant$}, ylabel={$H^{\type}(\quant)$},
  legend style={draw=none, fill=none, font=\small}, legend pos=south east,
  legend cell align=left]
\addplot[thick] coordinates {(0,0) (0.1,0) (0.1,1) (1,1)};
\addlegendentry{$\val=0.9$}
\addplot[thick, dashed] coordinates {(0,0) (0.4,0) (0.4,1) (1,1)};
\addlegendentry{$\val=0.6$}
\end{axis}
\end{tikzpicture}
\label{fig:quantile_dist_risk}
}
\hfill
\subfloat[private-budget agents]{
\begin{tikzpicture}
\begin{axis}[width=0.46\textwidth, height=0.36\textwidth,
  axis lines=left, xmin=0, xmax=1.02, ymin=0, ymax=1.08,
  xtick={0,0.5,1}, ytick={0,0.5,1},
  xlabel={quantile $\quant$}, ylabel={$H^{\type}(\quant)$},
  legend style={draw=none, fill=none, font=\footnotesize, at={(1.03,0.0)}, anchor=south east},
  legend cell align=left]
\addplot[thick, domain=0.055:0.786, samples=60] {min(1, 0.15/((3-sqrt(1+8*x))/2))};
\addlegendentry{$(\val,\wealth)=(0.9,0.15)$}
\addplot[thick, dashed, domain=0.28:0.48, samples=40] {min(1, 0.4/((3-sqrt(1+8*x))/2))};
\addlegendentry{$(\val,\wealth)=(0.6,0.40)$}
\addplot[thick, forget plot] coordinates {(0,0) (0.055,0) (0.055,0.1667)};
\addplot[thick, forget plot] coordinates {(0.786,1) (1,1)};
\addplot[thick, dashed, forget plot] coordinates {(0,0) (0.28,0) (0.28,0.6667)};
\addplot[thick, dashed, forget plot] coordinates {(0.48,1) (1,1)};
\end{axis}
\end{tikzpicture}
\label{fig:quantile_dist_budget}
}
\caption{The quantile distribution $H^{\type}$ when values (and budgets) are drawn independently from the uniform distribution on $[0,1]$. Panel~(a): for a risk-averse type $\type=(\val,\risk)$, $H^{\type}$ is a step function at the quantile $1-\valdist(\val)$ of the agent's value, regardless of her risk attitude $\risk$, and coincides with that of a risk-neutral agent with the same value. Panel~(b): for a private-budget type $\type=(\val,\wealth)$, $H^{\type}(\quant)=\min\{1,\sfrac{\wealth}{\price^{\quant}}\}$ for quantiles at which the market clearing price is at most $\val$, and $H^{\type}(\quant)=0$ otherwise.}
\label{fig:quantile_dist}
\end{figure}

This definition of quantiles also results in quantiles, for all distributions over non-linear preferences, being uniformly distributed.
Specifically, for any non-linear agent, the cumulative distribution function of quantile~$\quant$ is
\begin{align*}
\prob[\type\sim\typedist,\quant\sim H^{\type}]{\quant\leq z}
= \expect[\type\sim\typedist]{H^{\type}(z)}
= \expect[\type\sim\typedist]{\demand^{\type}(\price^z)}
= z
\end{align*}
where the last equality holds due to the definition of the market clearing price~$\price^z$.

\paragraph{Mechanisms in Quantile Space}
Given the definition of quantiles for non-linear agents, we define the per-unit pricing mechanisms for non-linear agents in \Cref{alg:meta theorem},
which first maps the types of all agents into quantiles based on \cref{def:quantile_nonlinear},
and allocates items to agents based on thresholds in the quantile space.\footnote{To simplify the exposition, we only state the payment rules without the consideration of the tie-breaking rules in the per-unit pricing mechanism for linear agents.
The details for constructing consistent tie-breaking rules are explained in \cref{apx:tie}.}

\begin{algorithm}[t]
\caption{Per-Unit Pricing Mechanisms for Non-linear Agents}
\label{alg:meta theorem}
\KwIn{Non-linear agents $\{(\typespace_i,
\typedist_i, \util_i)\}_{i\in \agents}$;
a per-unit pricing mechanism~$\mech$ for linear agents with threshold functions $\{Q_i\}_{i\in \agents}$ on quantiles.}
\vspace{2mm}
For each agent $i$ with private type $\type_i$,
map the type to a random quantile $\quanti$ according to cumulative distribution function $H^{\typei}_i(\cdot)$.
\vspace{2mm}

For each agent $i$, calculate quantile threshold as $\exquanti =
Q_i(\quant_{-i})$.
\vspace{2mm}

For each agent $i$, her allocation is $\alloci = 1$ if $\quanti \leq \exquanti$ and $\alloci=0$
otherwise.
The payment of agent $i$ is $\pricei = \price_{i}^{\exquanti}\cdot\demand^{\typei}(\price_{i}^{\exquanti})$ regardless of the allocation.
\end{algorithm}

\begin{theorem}\label{thm:same_payoff}
For any non-linear agents with price-posting revenue curves $\{\cumpayoff_i\}_{i\in \agents}$,
the per-unit pricing mechanism $\mech$ with threshold functions $\{Q_i\}_{i\in \agents}$ defined in \Cref{alg:meta theorem}
is DSIC and IIR for non-linear agents,
and its expected revenue is
\begin{align*}
\mech(\{\cumpayoff_i\}_{i\in \agents}) =
\expect[\{\quant_i\}_{i\in \agents} \sim {U[0,1]^n}]{\sum_{i\in \agents} \cumpayoff_i(Q_i(\quant_{-i}))}.
\end{align*}
\end{theorem}

\cref{thm:same_payoff} implies that for any set of linear agents and non-linear agents with the same price-posting revenue curves $\{\cumpayoff_i\}_{i\in \agents}$,
given any per-unit pricing mechanism~$\mech$ for linear agents,
there exists a per-unit pricing mechanism for non-linear agents with the same allocation rule in the quantile space that guarantees the same expected revenue.
Therefore, we can use $\mech(\{\cumpayoff_i\}_{i\in \agents})$ to represent the expected revenue of per-unit pricing mechanism $\mech$ given price-posting revenue curves $\{\cumpayoff_i\}_{i\in \agents}$
without the consideration of the detailed utility models of the agents.
The proof of \cref{thm:same_payoff} relies on the following alternative interpretation of \Cref{alg:meta theorem}.

\paragraph{Price-Posting Equivalence Interpretation}
For agent $i$ with type $\type_i$,
conditioning on threshold $\hat\quant_i = Q_i(\quant_{-i})$ which only depends on other agents' types,
her expected allocation over the randomness of her quantile~$\quanti$ is
\begin{align*}
 \expect{\indicator{\quanti \leq \hat{\quant}_i} \given \hat{\quant}_i,\typei} = H^{\typei}_i(\hat{\quant}_i)
 =
 \demand^{\typei}(\price_{i}^{\hat\quant_i})
\end{align*}
and her payment is deterministically $\price_{i}^{\hat{\quant}_i} \cdot \demand^{\typei}(\price_{i}^{\hat{\quant}_i})$.
Therefore, from the perspective of agent~$i$,
the constructed mechanism is equivalent to offering a market clearing price $\price_{i}^{\hat{\quant}_i}$
where the threshold quantile $\hat{\quant}_i$ only depends on other agents' types.

\begin{proof}[Proof of \cref{thm:same_payoff}]
Both DSIC and IIR
of mechanism $\mech$ are immediately guaranteed from the
price-posting equivalence interpretation.
Moreover, for any non-linear agent~$i$, conditional on the event that her threshold quantile is $\hat{\quant}_i=Q_i(\quant_{-i})$,
the expected revenue from agent $i$ in this event is $\cumpayoffi(\hat{\quant}_i)$.
Therefore, by linearity of expectation,
the expected revenue of mechanism $\mech$ is
\begin{equation}
\label{eq:pricing based payoff}
\mech(\{\cumpayoff_i\}_{i\in \agents}) =
\expect[\quant \sim {U[0,1]^n}]{\sum_{i\in \agents} \cumpayoff_i(Q_i(\quant_{-i}))}
\end{equation}
since the marginal distribution over quantile $\quanti$ is drawn from a uniform distribution $U[0,1]$ for all agents $i$.
\end{proof}

\subsection{Implementation via Demand-Based Auctions}
\label{sub:demand_reports}
\Cref{alg:meta theorem} is stated as a direct mechanism: agents report their types, and the seller computes each agent's demand from her reported type, which requires knowledge of her utility function. The price-posting equivalence interpretation shows that this knowledge is not needed. An agent's outcome depends on her type only through her demand at the per-unit price she faces, so the seller can ask each agent to report her demand curve instead of her type. The resulting demand-based auction is how we implement per-unit pricing mechanisms for non-linear agents: every per-unit pricing mechanism, described by its threshold functions, has a corresponding demand-based auction.

\begin{definition}[Demand-Based Auction]
\label{def:demand_auction}
Fix the agents' expected demand curves $\{\aggregateDemand_i\}_{i\in \agents}$, with market clearing prices $\price_i^{\quant}$, and a per-unit pricing mechanism for linear agents with threshold functions $\{Q_i\}_{i\in \agents}$. The corresponding \emph{demand-based auction} proceeds as follows.
\begin{enumerate}
\item Each agent $i$ reports a demand curve $\hat\demand_i$, i.e., a weakly decreasing function from per-unit prices to $[0,1]$.
\item For each agent $i$, the seller draws a quantile $\quanti$ with $\Pr[\quanti\leq\quant]=\hat\demand_i(\price_i^{\quant})$ for $\quant\in[0,1]$.\footnote{If $\hat\demand_i(\price_i^{1})<1$, the remaining probability is assigned to an event in which agent $i$ is not served; her quantile is then set to $1$ when the thresholds of the other agents are computed. Under truthful reporting, this event has probability zero, since $\aggregateDemand_i(\price_i^{1})=1$ implies $\demand_i^{\typei}(\price_i^{1})=1$ almost surely.}
\item Agent $i$ receives an item if and only if $\quanti\leq\exquanti$, where $\exquanti=Q_i(\quant_{-i})$, and pays $\price_i^{\exquanti}\cdot\hat\demand_i(\price_i^{\exquanti})$ regardless of the allocation.
\end{enumerate}
\end{definition}

\begin{proposition}
\label{prop:demand_auction}
In the demand-based auction, reporting her true demand curve $\hat\demand_i=\demand_i^{\typei}$ is a dominant strategy for every agent $i$ under every utility model, and the auction is IIR and ex post feasible. Under truthful reporting, its outcome coincides with that of \cref{alg:meta theorem}, and its expected revenue is $\mech(\{\cumpayoff_i\}_{i\in \agents})$ as given in \cref{thm:same_payoff}. Running the auction requires only the expected demand curves $\{\aggregateDemand_i\}_{i\in \agents}$, not the agents' utility functions or type distributions.
\end{proposition}
The proof is provided in Appendix~\ref{apx:demand_auction}.

Reports are simple in the two leading utility models. A risk-averse agent demands $\demand^{\type}(\price)=\indicator{\val\geq\price}$ (\cref{clm:risk_demand}), so her demand curve is described by a single number, her value. The demand-based auction then asks for a single bid, and her risk attitude need not be reported. A budget-constrained agent with value $\val$ and budget $\budget$ demands $\min\{1,\sfrac{\budget}{\price}\}$ at per-unit prices $\price\leq\val$ and nothing at higher prices, so her demand curve is described by a bid and a budget, which is the format used in online advertising. At the per-unit price she faces, she is allocated the largest share of the item that her budget affords. In both cases, \cref{prop:demand_auction} implies that truthful reporting of the value, or of the value and the budget, is a dominant strategy.

\subsection{Marginal Revenue Maximization}
\label{sub:mrm}
Different profiles of threshold functions $\{Q_i\}_{i\in \agents}$, that is, different allocation rules in quantile space, define different demand-based auctions, and a natural question is which of them maximizes revenue. By \cref{thm:same_payoff}, the revenue of each depends only on its thresholds and the price-posting revenue curves, so the question is the same as for linear agents with the same revenue curves.
\citet{BR-89} show that the marginal revenue maximization mechanism is revenue-optimal for linear agents.
Moreover, \citet{AFHH-13} show that for linear agents, it is without loss of optimality to focus on per-unit pricing mechanisms.
The expected revenue of a per-unit pricing mechanism $\mech$ with threshold functions $Q_i(\quant_{-i})$ given price-posting revenue curves $\{\cumpayoff_i\}_{i\in \agents}$ can be reformulated in terms of marginal revenue as
\begin{align*}
&\mech(\{\cumpayoff_i\}_{i\in \agents}) =
\expect[\quant \sim {U[0,1]^n}]{\sum_{i\in \agents} \cumpayoff_i(Q_i(\quant_{-i}))}
= \sum_{i\in \agents} \expect[\quant \sim {U[0,1]^n}]{\cumpayoff_i(Q_i(\quant_{-i}))}\\
&=\sum_{i\in \agents} \expect[\quant \sim {U[0,1]^n}]{\cumpayoff'_i(\quant_i) \cdot \hat{x}_i(\quant_i,\quant_{-i})}
=\expect[\quant \sim {U[0,1]^n}]{\sum_{i\in \agents} \cumpayoff'_i(\quant_i) \cdot \hat{x}_i(\quant_i,\quant_{-i})}
\end{align*}
where $\hat{x}_i(\cdot)$ is the allocation rule in quantile space induced by the threshold function $Q_i(\quant_{-i})$. That is, $\hat{x}_i(\quant_i,\quant_{-i}) = 1$ if and only if $\quant_i\leq Q_i(\quant_{-i})$.
The second and the last equalities hold by the linearity of expectation, and the third equality holds since
by the definition of  $\hat{x}_i(\cdot)$, we have
\begin{align*}
\cumpayoff_i(Q_i(\quant_{-i}))
= \int_0^1 \cumpayoff'_i(\quant_i) \cdot \hat{x}_i(\quant_i,\quant_{-i}) \dd\quant_i
= \expect[\quant_i \sim {U[0,1]}]{\cumpayoff'_i(\quant_i) \cdot \hat{x}_i(\quant_i,\quant_{-i})}.
\end{align*}
As interpreted by \citet{BR-89}, $\cumpayoff'_i(\quant_i)$ can be viewed as the marginal revenues for selling an item to agent $i$ with quantile $\quant_i$.
To maximize the expected revenue, the marginal revenue maximization mechanism ($\MRM$) allocates the items to agents with the highest marginal revenues $\cumpayoff'_i(\quant_i)$.
In the case where $\cumpayoff_i$ is concave for any agent $i$ and hence $\cumpayoffi'(\quanti)$ is monotone in $\quanti$, the resulting mechanism is a per-unit pricing mechanism and maximizes the expected revenue of the seller.
Since price-posting revenue curves need not be concave, we also use their concave hulls.
\begin{definition}
\label{def:ironed}
The \emph{ironed price-posting revenue curve} $\payoffHull$ is the concave hull of the price-posting revenue curve $\cumpayoff$.
\end{definition}
Since $\cumpayoff$ is a single-dimensional function, $\payoffHull(\quant)$ can be obtained by randomizing over at most two per-unit prices, and it is the optimal revenue from posting a randomized per-unit price subject to the ex ante supply constraint $\quant$.
In the case that $\cumpayoffi(\quanti)$ is not concave in $\quanti$,
we can apply the ironing trick in \cite{mye-81}
by considering allocations that maximize the ironed marginal revenues
$\sum_{i\in \agents} \payoffHull'_i(\quant_i) \cdot \hat{x}_i(\quant_i,\quant_{-i})$.

For non-linear agents, since restricting attention to per-unit pricing mechanisms may not be without loss of generality, the marginal revenue maximization mechanism may not be optimal among all possible mechanisms.
Nonetheless, by adopting the same argument as for linear agents, the marginal revenue maximization mechanism maximizes the seller's expected revenue among the subclass of per-unit pricing mechanisms regardless of the agents' utility models.
Moreover, in \cref{sec:reduction}, we will show that
marginal revenue maximization is approximately optimal for broad classes of non-linear utilities.

Finally, note that the expected revenue of the marginal revenue maximization mechanism only depends on the price-posting revenue curves, not the detailed utility models.
Therefore, we denote
the revenue of marginal revenue maximization $(\MRM)$ for agents
with price-posting revenue curves $\{\cumpayoff_i\}_{i\in \agents}$ as
$\MRM(\{\cumpayoff_i\}_{i\in \agents})$.

\subsection{Selling to Risk-Averse and Budget-Constrained Agents}
\label{sub:illustrations}
The description of per-unit pricing mechanisms in the quantile space is beneficial for understanding the economics behind revenue maximization, and will be useful for our subsequent analysis that shows its optimality or approximate optimality.
In practice, it is also useful to know what these mechanisms look like in the value space. \Cref{sub:demand_reports} gives a general answer that applies to every utility model: any per-unit pricing mechanism can be run as a demand-based auction, in which agents report entire demand curves. This subsection instead exploits the structure of the two leading utility models, and shows that marginal revenue maximization then takes a simpler and familiar form. To keep the formats transparent, we consider a single item and i.i.d.\ agents whose price-posting revenue curve $\cumpayoff$ is concave; with $m$ units, the second-price auction below is replaced by the $(m+1)$-th price auction, and nothing else changes.

\paragraph{Risk-Averse Agents: Second-Price Auction with a Monopoly Reserve}
We first consider the case where agents are risk averse. We show that any risk-averse agent, when facing a per-unit price, will behave the same as a linear agent. The proof of this claim is provided in \cref{apx:proof rev risk}.
\begin{claim}\label{clm:risk_demand}
For any risk-averse agent with type $\type=(\val,\risk)$ and for any per-unit price~$\price$, the expected utility of the agent is maximized at allocation $z\in\{0, 1\}$.
\end{claim}
Therefore, the price-posting revenue curves of risk-averse agents coincide with the price-posting revenue curves of linear agents with the same value distributions. Moreover, the per-unit pricing mechanisms for risk-averse agents also coincide with linear agents in the value space.
For example, for i.i.d.\ linear agents with regular value distributions, the marginal revenue maximization mechanism can be implemented as the second-price auction with monopoly price as the reserve \citep{mye-81}.
Our observations above imply that for risk-averse agents with i.i.d.\ value distributions and arbitrary distribution on the risk attitude, the marginal revenue maximization mechanism can also be implemented as the second-price auction with monopoly price as the reserve.

\paragraph{Budget-Constrained Agents: Second-Price Auction with Randomized Budget Caps}
The implementation of marginal revenue maximization is trickier for other non-linear utilities such as budget constraints. The main reason is that a budget-constrained agent behaves differently from a linear agent when facing per-unit prices.
Nevertheless, marginal revenue maximization remains a second-price auction with a monopoly reserve once the bids are suitably randomized.

\begin{proposition}
\label{prop:budget_cap}
Consider selling a single item to i.i.d.\ budget-constrained agents whose price-posting revenue curve is concave, and let $\price^\star$ be the monopoly price, i.e., the market clearing price at the quantile that maximizes the price-posting revenue curve. For each agent $i$ with reported value $\val_i$ and budget $\wealth_i$, let
\begin{align*}
b_i=\min\left\{\val_i,\ \frac{\wealth_i}{U_i}\right\},
\end{align*}
where $U_i$ is drawn independently from the uniform distribution on $[0,1]$. Then marginal revenue maximization is equivalent to the following second-price auction with reserve $\price^\star$ on the bids $\{b_i\}_{i\in \agents}$: agent $i$ faces the price $\rho_i=\max\{\price^\star,\max_{j\neq i}b_j\}$, receives the item if and only if $b_i\geq \rho_i$, and pays $\min\{\rho_i,\wealth_i\}$ if $\rho_i\leq\val_i$ and nothing otherwise, regardless of the allocation.\footnote{Ties, which occur with probability zero when values are continuously distributed, are broken uniformly at random.}
\end{proposition}
The proposition is easy to verify from the demand-based auction (\cref{def:demand_auction}). Since the agents are i.i.d., they have the same market clearing price function $\price^{\quant}$, and setting $b_i=\price^{\quanti}$ turns marginal revenue maximization into a second-price auction with reserve $\price^\star$. Moreover, for every per-unit price $\price\geq 0$, the agent's probability of winning at price $\price$, $\Pr[b_i\geq\price]$, equals her demand $\min\{1,\sfrac{\wealth_i}{\price}\}\cdot\indicator{\val_i\geq\price}$, which is also $\Pr[\min\{\val_i,\sfrac{\wealth_i}{U_i}\}\geq\price]$.

Each bid is the agent's value, capped by her budget scaled up by a random factor $\sfrac{1}{U_i}\geq 1$; we call this a \emph{randomized budget cap}. Since the budget is only scaled up, never down, an agent whose budget covers her price $\rho_i\leq\val_i$ wins for sure and pays $\rho_i$, exactly as in the standard second-price auction, and an agent whose budget exceeds her value simply bids her value. An agent whose budget falls short of her price $\rho_i\leq\val_i$ instead pays her budget $\wealth_i$ for the winning probability $\sfrac{\wealth_i}{\rho_i}$, which is the share of the item her budget affords. She pays whether or not she wins, since paying $\rho_i$ upon winning would exceed her budget. Like the second-price auction with a monopoly reserve for linear agents, the auction is prior-independent except for the reserve: each cap uses only the agent's own report and an independent uniform draw.

This auction also improves on the standard second-price auction with a reserve, in which a budget-constrained agent optimally bids $\min\{\val_i,\wealth_i\}$, as if she were a linear agent whose value is capped by her budget. Compared with such a linear agent, a budget-constrained agent demands weakly more at every price, since at a price $\price\in(\wealth_i,\val_i]$ she still buys the share $\sfrac{\wealth_i}{\price}$ that her budget affords. Her price-posting revenue curve is therefore weakly higher, and by \cref{thm:same_payoff}, marginal revenue maximization yields weakly higher revenue than the second-price auction with any reserve. The gain can be unbounded, even with a single agent.
\begin{example}[Second-price auction with a budget-constrained agent]
\label{ex:spa_budget}
Consider a single agent with value $1$ and a budget $\wealth$ drawn from the distribution with $\Pr[\wealth\geq x]=\sfrac{\epsilon}{x}$ for $x\in[\epsilon,1]$. In the second-price auction with reserve $r$, she bids $\wealth$ and wins if and only if $\wealth\geq r$, so the revenue is $r\cdot\Pr[\wealth\geq r]\leq\epsilon$ for every $r$. Under marginal revenue maximization, she faces the monopoly price $1$ and pays her budget for the winning probability $\wealth$, so the revenue is $\expect{\wealth}=\epsilon\,(1+\ln\sfrac{1}{\epsilon})$. The ratio $1+\ln\sfrac{1}{\epsilon}$ is unbounded as $\epsilon\to0$.
\end{example}

Hence, in both leading cases, marginal revenue maximization is a second-price auction with a monopoly reserve. Budget constraints change only how bids are formed and require that lotteries be paid for up front.
The up-front payment is the source of the gain: with prices paid only upon winning, a budget-constrained agent would buy if and only if $\price\leq\min\{\val,\wealth\}$, like a linear agent with value $\min\{\val,\wealth\}$, and the auction would reduce to the standard second-price auction of \cref{ex:spa_budget}.

\section{Performance Guarantees}
\label{sec:guarantees}
\Cref{sec:implementation} constructs a mechanism whose revenue depends only on the agents' expected demand curves. In this section, we evaluate its performance, which depends on how much the seller knows, and we consider two informational extremes.

At one extreme, the seller observes only each agent's expected demand curve. She then cannot distinguish a risk-averse agent from a budget-constrained or risk-neutral agent who purchases the same amount at every price. Theorem~\ref{thm:robust_optimality} in Section~\ref{sec:robust} shows that, under this ambiguity, marginal revenue maximization is max-min optimal.

At the other extreme, the seller knows each agent's utility function and type distribution and may offer arbitrary menus of lotteries. In this case, the mechanism is generally not optimal. Theorem~\ref{thm:meta thm expected utility} in Section~\ref{sec:reduction} bounds its loss by a factor $\gamma\zeta$, where $\gamma$ is the approximation factor of the same mechanism for linear agents and $\zeta$ measures how far the agents' non-linearity pushes the optimal mechanism beyond posted pricing.

These two results bracket all intermediate information structures. By Theorem~\ref{thm:same_payoff}, the revenue of the mechanism depends only on the demand curves, which the seller observes in all cases, whereas the optimal revenue with partial information is at most the optimal revenue with full information. Hence, under any intermediate information structure, the mechanism is within a factor $\gamma\zeta$ of the optimal revenue.

This gives $\zeta$ an economic interpretation: it bounds the \emph{value to the seller of learning the details of non-linearity}, that is, of learning the agents' utility functions and type distributions beyond what their demand curves reveal, such as whether agents are risk averse or budget constrained, to what degree, and how these features are jointly distributed with values. When $\zeta$ is close to one, such information is of little value, and the seller loses little by running a demand-based auction. When $\zeta$ is large, the details are valuable.

Section~\ref{sec:revenue close} bounds $\zeta$ for canonical utility models and shows that, under standard assumptions on value distributions, the value of detail is bounded by a constant for risk aversion and for budgets independent of values, but can be unbounded when budgets are correlated with values. Section~\ref{sec:example} shows that in numerical examples, the actual losses are much smaller than the worst-case bounds.

\subsection{Robustness to Unobserved Details}
\label{sec:robust}

In this subsection, we focus on a robust setting wherein the seller possesses only the knowledge of the agents' expected demand curves, lacking detailed information regarding the non-linear utility function or the type distribution. We show that the marginal revenue maximization mechanism (see \cref{sub:mrm}) is robustly optimal in such environments.

Given any expected demand curve $\aggregateDemand(\price)$ that is a non-increasing function,\footnote{We focus on non-increasing demand curves since it is without loss under the ordinary goods assumption. }
let $\feasibleUF(\aggregateDemand)$ be the set of feasible agent utility models $(\typespace, \typedist, \util)$
such that
for any $\price\geq 0$,
\begin{align*}
\aggregateDemand(\price) = \expect[\type\sim\typedist]{\demand^{\type}(\price; \util)}
\end{align*}
where $\demand^{\type}(\price; \util)$ is the demand curve of type $\type\in\typespace$ induced by utility function $\util$.
Note that the expected demand curve uniquely pins down the price-posting revenue curve $\cumprice$ regardless of the utility models.

For any mechanism $\mech$, let $\mech(\lbr{(\typespace_i, \typedist_i, \util_i)}_{i\in \agents})$ be the expected revenue of mechanism $\mech$ given a profile of utility models $\lbr{(\typespace_i, \typedist_i, \util_i)}_{i\in \agents}$.
Since the seller commits to the mechanism before learning the utility models, $\mech$ here is an indirect mechanism, consisting of a message space for each agent and an outcome function, and its expected revenue given a profile of utility models is evaluated at a Bayes--Nash equilibrium of the game it induces.
In this subsection, we assume that the seller can only observe the expected demand curve $\aggregateDemand_i(\price)$ for all agents $i$ without knowing the detailed utility models $(\typespace_i, \typedist_i, \util_i)$.
The justification is that the detailed utility models are often hard to estimate based on the historical purchase data from consumers for practical applications. However, it is easier to estimate the demand curves based on posted pricing \citep{bajari2015machine}.

When the seller is ignorant of the detailed utility models,
the robust approach considers the objective of maximizing the expected revenue over the worst-case utility models that are consistent with the observed expected demand curves.
That is, the seller chooses a mechanism~$\mech$ to maximize
\begin{align*}
\min_{(\typespace_i, \typedist_i, \util_i) \in \feasibleUF(\aggregateDemand_i), \forall i\in \agents}
\mech(\lbr{(\typespace_i, \typedist_i, \util_i)}_{i\in \agents}).
\end{align*}

\begin{restatable}{theorem}{robustOptimality}
\label{thm:robust_optimality}
For any profile of non-increasing expected demand curves $\{\aggregateDemand_i\}_{i\in \agents}$, if the seller does not observe the detailed profile of utility models, the marginal revenue maximization mechanism, run as a demand-based auction, is the max-min revenue-optimal mechanism,
i.e.,
\begin{align*}
\MRM \in \argmax_{\mech} \min_{(\typespace_i, \typedist_i, \util_i) \in \feasibleUF(\aggregateDemand_i), \forall i\in \agents}
\mech(\lbr{(\typespace_i, \typedist_i, \util_i)}_{i\in \agents}).
\end{align*}
\end{restatable}
The proof, given in Appendix~\ref{apx:robust}, combines two ingredients. First, marginal revenue maximization, run as a demand-based auction, obtains the same revenue under every utility model in $\feasibleUF(\aggregateDemand)$ (\cref{thm:same_payoff,prop:demand_auction}). Second, the linear utility model belongs to $\feasibleUF(\aggregateDemand)$, and for linear agents, no mechanism outperforms marginal revenue maximization \citep{BR-89}.

Since the argument requires only that the ambiguity set contain the linear utility model, \cref{thm:robust_optimality} continues to hold when the seller knows more about the agents' utilities. For example, if the seller knows that the agents are weakly risk averse, which includes risk neutrality, and knows their value distributions but not their risk attitudes, marginal revenue maximization remains max-min optimal; in the symmetric environments of \cref{sub:illustrations}, it is the second-price auction with a monopoly reserve. The same holds if the seller knows that the agents are risk neutral with budgets, which need not bind, but observes the joint distribution of values and budgets only through the expected demand curves; there, it is the second-price auction with randomized budget caps.

\Cref{thm:robust_optimality} does not mean, however, that the seller can ignore non-linearity and run an auction designed for linear agents. In the second-price auction with a monopoly reserve, which is optimal for i.i.d.\ linear agents with regular value distributions, a budget-constrained agent is excluded whenever her price exceeds her budget, and \cref{ex:spa_budget} shows that the seller can then lose almost all of the revenue that the demand-based auction guarantees. Hence, while a seller who observes only demand curves should avoid complex lotteries, she must still account for non-linearity by selling through a demand-based auction.

\subsection{Approximation with Known Details}
\label{sec:reduction}
For non-linear agents, the per-unit pricing mechanisms proposed in \cref{sec:implementation} in general are not optimal for revenue maximization if the seller has precise information regarding the environments.
In this subsection, we show that those mechanisms are approximately optimal for broad classes of non-linear utilities
by providing a reduction framework that extends the approximation bounds of
per-unit pricing mechanisms for linear agents
to non-linear agents.
An important application of our general reduction framework is that the robustly optimal mechanism, marginal revenue maximization,
is approximately optimal against the optimal detail-dependent mechanism.

Specifically, our reduction framework (\Cref{thm:meta thm expected utility}) shows that if non-linear agents are $\zeta$-resemblant, restricting attention to per-unit pricing mechanisms, which are robust to the details of the non-linear model, loses at most a factor $\zeta$ beyond the loss already incurred for linear agents.

\subsubsection{Optimal Revenue Curves and Resemblance}
\label{sub:single_curve}
In this subsection, we analyze the single-agent problem with ex ante supply constraints.
We define the \emph{optimal revenue curves} as follows.
The main difference between the optimal revenue curves and the price-posting revenue curves (\cref{def:ppcurve}) is that the former do not restrict attention to mechanisms that post per-unit prices.

\begin{definition}
The \emph{optimal revenue curve} $\payoffcurve(\quant)$
is a mapping from any ex ante supply constraint $\quant$
to the optimal ex ante revenue for the single-agent problem
that in expectation sells the item with probability~$\quant$,
i.e.,
\begin{align*}
\payoffcurve(\quant) = \max_{\distmech \text{ \rm is BIC, IIR}}
\quad & \expect[(\alloc,\price)\sim\distmech]{\expect[\type\sim\typedist]{\RealizedPayoff(\type,\alloc(\type),\price(\type))}}\\
\text{s.t.} \qquad & \expect[(\alloc,\price)\sim\distmech]{\expect[\type\sim\typedist]{\alloc(\type)}} = \quant.
\end{align*}
\end{definition}

\begin{lemma}
\label{fact:concave}
The optimal revenue curve is
concave.
\end{lemma}
This lemma is immediate since in this case the space of incentive compatible mechanisms is closed under convex combination,
and the revenue is linear in the space of mechanisms.
Next we review the relation between
the optimal revenue curves and
the concave hull of the price-posting revenue curves
for linear agents.
Recall from \cref{def:ironed} that this concave hull, the ironed price-posting revenue curve, is the optimal revenue from posting a randomized per-unit price.

\begin{lemma}[\citealp{BR-89}]
\label{lem:price-posting revenue curve to ex ante revenue curve}
The optimal revenue curve
$\revcurve$ of a linear agent is equal to
her ironed price-posting revenue curve $\concaveHull$.
\end{lemma}

In general, for non-linear agents,
the optimal revenue curves and the ironed price-posting revenue curves are not equivalent,
and the mechanism maximizing the ex ante revenue
is more complicated and can extract strictly higher revenue
than the optimal price-posting mechanism
or randomizations over price-posting mechanisms.
We quantify the extent to which a non-linear agent resembles a linear agent
based on the maximum revenue gap between posted pricing and optimal mechanisms
in a single-agent environment with ex ante supply constraints.

\paragraph{$\zeta$-Resemblance}
We introduce \emph{$\zeta$-resemblance}
of an agent to measure how closely a non-linear agent resembles a linear agent whose ironed price-posting revenue curve
and optimal revenue curve are equal.
This definition is illustrated in \Cref{fig:resemblance}.\footnote{We consider the monotone extension of the ironed price-posting revenue curve, as we allow for underselling when analyzing the optimal revenue from price-posting mechanisms under ex ante supply constraints. Our condition is weaker than requiring that $\payoffHull(\quant)
\geq \sfrac{1}{\zeta} \cdot
\payoffcurve(\quant)$ for all $\quant\in[0,1]$.}
Note that based on our definition, \cref{lem:price-posting revenue curve to ex ante revenue curve} implies that a linear agent is $1$-resemblant.

\begin{definition}[$\zeta$-resemblance]
\label{def:zeta closeness}
An agent's ironed price-posting revenue curve $\payoffHull$ is
\emph{$\zeta$-resemblant} to her optimal revenue curve $\payoffcurve$,
if for all $\quant\primed \in [0, 1]$,
there exists $\quant \leq \correspondingquant$
such that
$\payoffHull(\quant)
\geq \sfrac{1}{\zeta} \cdot
\payoffcurve(\correspondingquant)$.
Such an agent is \emph{$\zeta$-resemblant}.
\end{definition}

\begin{figure}[t]
\centering
\setlength{\unitlength}{1cm}
\begin{tikzpicture}[scale = 0.6]

\draw (-0.2,0) -- (12.5, 0);
\draw (0, -0.2) -- (0, 6);

\draw [thick] plot [smooth, tension=0.6] coordinates { (0,0) (0.6, 1.4) (1.4,2) (3.7,2.9) (5, 3)};
\draw [thick] plot [smooth, tension=0.6] coordinates {(5, 3) (8.5,2) (12,0)};

\draw plot [smooth, tension=0.6] coordinates { (0,0) (2,4) (8, 5) (12,0)};

\draw [dashed] plot [smooth, tension=0.6] coordinates { (0,0) (2, 2.2) (8, 2.75) (12, 0) };

\begin{scope}[line width=0.6mm]
\draw [gray, dashed] plot [smooth, tension=0.6] coordinates { (0,0) (0.6, 1.4) (1.4,2) (3.7,2.9) (5, 3)};
\draw [gray, dashed] plot (5, 3) -- (12, 3);
\end{scope}

\draw (0, -0.7) node {$0$};
\draw (12, -0.7) node {$1$};

\draw (3.6, 5.4) node {$\revcurve$};
\draw (8, 1.5) node {$\concaveHull$};

\draw (4, 2) node {$\sfrac{\revcurve}{\zeta}$};

\end{tikzpicture}

\caption{The thin dashed curve is the optimal revenue curve divided by~$\zeta$,
and the thick gray dashed curve is the monotone extension of the ironed price-posting revenue curve. The definition of $\zeta$-resemblance is satisfied if the thick dashed curve is always above the thin dashed curve.}
\label{fig:resemblance}
\end{figure}

Our definition implies that for a $\zeta$-resemblant agent,
given any ex ante supply constraint~$\quant$,
there exists a distribution over per-unit prices with expected sale probability at most~$\quant$ that
achieves at least a $\frac{1}{\zeta}$ fraction of the optimal revenue with this ex ante constraint.\footnote{This is equivalent to an environment of selling items to a continuum of agents with a supply constraint,
where $\zeta$-resemblance measures the worst-case gap (over all possible supplies) between optimal mechanisms and
posted pricing mechanisms for non-linear agents. }
Since posted pricing mechanisms are optimal mechanisms
for a linear agent,
smaller~$\zeta$ implies that non-linear agents
resemble linear agents better.

\subsubsection{Reduction Framework}
\label{sub:meta_theorem}

The reduction framework in this section relies on
the following technique of ex ante relaxation
for bounding the optimal revenue of the seller.

\paragraph{Ex Ante Relaxation}
In general for non-linear agents,
it is challenging to characterize the optimal mechanism and quantify the optimal revenue of the seller.
Thus, we adopt the benchmark of \emph{ex ante relaxation} \citep{yan-11},
which serves as an upper bound of the optimal revenue.
For selling $m$ units of identical items,
a sequence of ex ante quantiles $\{\quanti\}_{i\in \agents}$
is ex ante feasible
if $\sum_i\quanti \leq m$.
We denote the set of ex ante feasible quantiles by $\exfeasible$.
The optimal ex ante revenue given the
optimal revenue curves $\{\payoffcurve_i\}_{i\in \agents}$ is
$$\exanterelax(\{\payoffcurve_i\}_{i\in \agents})
= \max\limits_{\{\quanti\}_{i\in \agents}
\in \exfeasible}
\sum\nolimits_{i\in \agents} \payoffcurve_i(\quanti).$$
Note that $\exanterelax(\{\payoffcurve_i\}_{i\in \agents})$ is an upper bound on the optimal revenue since any feasible mechanism must satisfy the ex ante feasibility constraints,
and the ex ante revenue optimizes the seller's expected revenue under such constraints.
Moreover, the ex ante revenue is uniquely determined by the optimal revenue curves regardless of the detailed utility models of the agents.

\paragraph{Reduction of Approximation Guarantees}
Now we present the following
theorem:
a reduction framework that
extends the approximation guarantees
of every per-unit pricing mechanism
for linear agents to non-linear agents.

\begin{theorem}[Reduction Framework]
\label{thm:meta thm expected utility}
For any per-unit pricing mechanism $\mech$
and any set of non-linear
agents with price-posting revenue curves
$\{\cumpayoff_i\}_{i\in \agents}$
and optimal revenue curves $\{\payoffcurve_i\}_{i\in \agents}$,
if
\begin{itemize}
\item all non-linear agents are $\zeta$-resemblant; and
\item mechanism $\mech$ is a $\gamma$-approximation to the ex ante relaxation for linear agents with the same price-posting revenue curves, i.e., $\mech(\{\cumpayoff_i\}_{i\in \agents}) \geq
\sfrac{1}{\gamma} \cdot \exanterelax(\{\payoffHull_i\}_{i\in \agents})$,
\end{itemize}
then mechanism $\mech$
for non-linear agents (from \cref{alg:meta theorem})
is a $\gamma\,\zeta$-approximation
to the ex ante relaxation,
i.e.,
$\mech(\{\cumpayoff_i\}_{i\in \agents}) \geq
\sfrac{1}{\gamma\,\zeta} \cdot \exanterelax(\{\payoffcurve_i\}_{i\in \agents})$.
\end{theorem}

Before the proof of \cref{thm:meta thm expected utility},
we first explain the conditions in \cref{thm:meta thm expected utility},
and in particular why we refer to
condition
$\mech(\{\cumpayoff_i\}_{i\in \agents}) \geq
\sfrac{1}{\gamma} \cdot \exanterelax(\{\payoffHull_i\}_{i\in \agents})$
as the approximation guarantees for linear agents.

Note that for both per-unit pricing mechanism $\mech$ and the ex ante relaxation,
the expected revenue of the seller only depends on the price-posting revenue curves and optimal revenue curves respectively,
not on the details of the utility models.
Therefore, for any set $\cA$ of non-linear agents with price-posting revenue curves $\{\cumpayoff_i\}_{i\in \agents}$,
we can define a \emph{linear agents analog} as a set $\cA^L$ of agents with the same price-posting revenue curves.\footnote{For revenue maximization,
the linear agents analog is well defined since
the price-posting revenue curve $\cumpayoff(\quant)$ of a linear agent uniquely pins down
her value distribution as $\val(\quant) = \frac{\cumpayoff(\quant)}{\quant}$.}
By \cref{thm:same_payoff}, the expected revenues of mechanism $\mech$ are the same for both non-linear agents $\cA$
and the linear agents analog $\cA^L$,
which equals $\mech(\{\cumpayoff_i\}_{i\in \agents})$ defined in \Cref{eq:pricing based payoff}.
Moreover, by \cref{lem:price-posting revenue curve to ex ante revenue curve},
the optimal revenue curves for the linear agents analog are
$\{\payoffHull_i\}_{i\in \agents}$,
and hence the optimal ex ante revenue for the linear agents analog
is $\exanterelax(\{\payoffHull_i\}_{i\in \agents})$.
Therefore, the condition that
$\mech(\{\cumpayoff_i\}_{i\in \agents}) \geq
\sfrac{1}{\gamma} \cdot \exanterelax(\{\payoffHull_i\}_{i\in \agents})$
can be interpreted as the approximation ratio of per-unit pricing mechanism $\mech$ for the linear agents analog
with respect to its ex ante relaxation benchmark.\footnote{This condition is necessary for our reduction framework since linear utility is a special case of non-linear utilities.}

Another condition in \cref{thm:meta thm expected utility} is the requirement that all non-linear agents are $\zeta$-resemblant.
This condition is also necessary for per-unit pricing mechanisms to be approximately optimal for non-linear agents.
Without this assumption, even in single-agent environments with supply constraints, there exist mechanisms with more complicated price discrimination schemes that significantly outperform per-unit pricing mechanisms.

\begin{proof}[Proof of \cref{thm:meta thm expected utility}]
Given the assumption that per-unit pricing mechanism $\mech$ is approximately optimal for linear agents,
i.e., $\mech(\{\cumpayoff_i\}_{i\in \agents}) \geq
\sfrac{1}{\gamma} \cdot \exanterelax(\{\payoffHull_i\}_{i\in \agents})$,
it is sufficient to show that
$\exanterelax(\{\payoffHull_i\}_{i\in \agents}) \geq
\sfrac{1}{\zeta} \cdot \exanterelax(\{\payoffcurve_i\}_{i\in \agents})$
when all non-linear agents are $\zeta$-resemblant.

Let $\{\quanti\primed\}_{i\in \agents} \in \exfeasible$ be the profile of optimal ex ante quantiles
for optimal revenue curves $\{\payoffcurve_i\}_{i\in \agents}$.
Since the ironed price-posting revenue curves
$\{\payoffHull_i\}_{i\in \agents}$ are $\zeta$-resemblant
to the optimal revenue curves $\{\payoffcurve_i\}_{i\in \agents}$,
there exists a sequence of quantiles $\{\quanti\}_{i\in \agents}$
such that for any agent $i$,
$\quanti \leq \quanti\primed$
and $\payoffHull_{i}(\quanti)
\geq \sfrac{1}{\zeta} \cdot \payoffcurve_{i}(\quanti\primed)$.
Note that
$\{\quanti\}_{i\in \agents}$ is also ex ante feasible.
Therefore,
\begin{equation*}
\exanterelax(\{\payoffcurve_i\}_{i\in \agents})
= \sum_{i\in \agents} \payoffcurve_i(\quanti\primed)
\leq \zeta \cdot \sum_{i\in \agents}
\payoffHull_i(\quanti)
\leq \zeta \cdot
\exanterelax(\{\payoffHull_i\}_{i\in \agents}).
\qedhere
\end{equation*}
\end{proof}

Note that our reduction framework (\cref{thm:meta thm expected utility}) is useful for practical applications only if both conditions in the statement are satisfied for commonly adopted per-unit pricing mechanisms and non-linear utilities. \Cref{sub:reduction_application,sec:revenue close} verify these conditions for marginal revenue maximization and for canonical non-linear utility models, respectively.

\subsubsection{Approximations of Marginal Revenue Maximization}
\label{sub:reduction_application}

Although marginal revenue maximization is revenue optimal for linear agents \citep{BR-89,AFHH-13}, it falls short of the ex ante relaxation, an upper bound on the optimal revenue, by a factor of at most $\sfrac{1}{(1-\frac{1}{\sqrt{2\pi m}})}$ when selling $m$ units of identical items \citep{yan-11}.
Combining this with \cref{thm:meta thm expected utility}, we obtain the following corollary.

\begin{corollary}\label{cor:mrm}
For selling $m$ units of identical items to any set of non-linear
agents with price-posting revenue curves
$\{\cumpayoff_i\}_{i\in \agents}$
whose ironed curves $\{\payoffHull_i\}_{i\in \agents}$ are $\zeta$-resemblant to their optimal revenue curves $\{\payoffcurve_i\}_{i\in \agents}$,
marginal revenue maximization is a
$\sfrac{\zeta}{(1-\frac{1}{\sqrt{2\pi m}})}$-approximation to the ex ante relaxation for non-linear agents,
i.e., $\MRM(\{\cumpayoff_i\}_{i\in \agents}) \geq
\frac{1}{\zeta} \cdot (1-\frac{1}{\sqrt{2\pi m}})
\cdot \exanterelax(\{\payoffcurve_i\}_{i\in \agents})$.
\end{corollary}

The bound in \cref{cor:mrm} has two parts. The factor $1-\frac{1}{\sqrt{2\pi m}}$ captures the loss that marginal revenue maximization incurs even for linear agents, because the ex ante relaxation imposes the supply constraint only in expectation. This loss vanishes as the number of units $m$ grows, so it is negligible in large markets. There, the approximation factor is essentially $\zeta$, which measures how closely the non-linear agents resemble linear agents and which we bound for canonical utility models next.

\subsection{The Value of Detail in Canonical Models}
\label{sec:revenue close}
In this subsection, we instantiate our reduction framework by showing that the resemblance property is indeed satisfied for canonical models of non-linear utility.
We show that both risk-averse agents and private-budget agents resemble linear agents under mild assumptions on the value distributions. In \cref{sec:revenue correlated} and \cref{apx:necessity_rev_resemble_assuption}, we show that those assumptions are necessary for per-unit pricing mechanisms to be approximately optimal even in single-agent settings without ex ante supply constraints.

As discussed at the beginning of Section~\ref{sec:guarantees}, the revenue of our mechanism is determined by the agents' expected demand curves, which the seller observes under every information structure we consider, whereas the optimal benchmark depends on her information. Hence, $\zeta$ bounds the gain the seller can obtain by learning, beyond the demand curves, the details of \emph{why} her agents deviate from linearity, such as their risk attitudes or budgets and how these are jointly distributed with values. We refer to this gain as the \emph{value of detail}.

Before turning to the formal results, we note that the value of detail is analyzed in terms of its maximum possible magnitude: each result bounds how much a seller could gain from learning the details, across all environments in a given class. The exact constants are therefore of secondary importance. What matters is how the bounds compare across classes, since the comparison reveals which features of utility models can make the value of detail large. For example, when values are regularly distributed, we show that the value of detail is at most $3$ if budgets are independent of values, but it can be arbitrarily large if budgets and values are correlated. This contrast indicates that the correlation between values and budgets is a salient feature for the value of detail, whereas the mere presence of budget constraints is not.

Since we will focus on the single-agent analysis in this subsection, we will drop the subscript representing the agent in all notations.

\subsubsection{Risk-Averse Agent}
\label{sec:revenue risk}
For a risk-averse agent, her private type is represented by a pair $\type=(\val, \risk)$ where $\val$ is the private value for the item and $\risk$ represents the agent's risk attitude. Note that we allow arbitrary correlations between the values and risk attitudes, and we do not impose structures on the risk attitudes such as constant absolute risk aversion (CARA) utilities.

Recall that $\valdist$ is the marginal distribution over values. We assume that distribution~$\valdist$ satisfies the monotone hazard rate (MHR) condition. This condition is satisfied by many common distributions in practice such as uniform, exponential, and Gaussian distributions.

\begin{definition}\label{def:mhr}
A distribution $\valdist$ with density function $\dens$ satisfies the \emph{monotone hazard rate (MHR)} condition if the hazard rate function $h(\val) \triangleq \frac{\dens(\val)}{1-\valdist(\val)}$ is non-decreasing in~$\val$.
\end{definition}

Essentially, \cref{def:mhr} implies that the tail probabilities in the value distribution are small, i.e., smaller than the exponential distribution \citep[see, e.g.,][]{AB-18}.
That is, the purpose of this assumption is to rule out fat-tail distributions.

\begin{restatable}{proposition}{riskaversebound}
\label{thm:risk-averse bound}
A risk-averse agent is $e$-resemblant if her marginal value distribution satisfies the MHR condition.
\end{restatable}

The formal proof of the proposition can be found in \Cref{apx:proof rev risk}. At a high level, the proof consists of a two-step argument. First, we argue that for every mechanism, its expected revenue is upper bounded by the expected allocated value due to the individual rationality constraint. One of the main technical challenges for analyzing risk-averse agents is that the revenue-optimal mechanism is hard to characterize even for a single agent under specific structures such as CARA. The first step bypasses this challenge, since it now suffices to compare the optimal expected allocated value achieved among all mechanisms with the optimal expected revenue among posted pricing mechanisms, in which we prove the gap is at most $e$ under the MHR condition.

\subsubsection{Budget-Constrained Agent}
\label{sec:revenue budget}
For a private-budget agent, her private type is represented by a pair $\type=(\val, \wealth)$ where $\val$ is the private value for the item and $\wealth$ represents the agent's private budget constraint. We assume that values and budgets are independent; \cref{sec:revenue correlated} shows that this assumption cannot be dropped.

In addition to the independence assumption, we also impose the standard regularity condition on the marginal value distribution $\valdist$ over values. This condition is strictly weaker than the MHR condition and can even include long-tail distributions such as Pareto distributions with a shape parameter at least $1$.

\begin{definition}\label{def:regular}
A distribution $\valdist$ with density function $\dens$ is \emph{regular} if the virtual value function $\phi(\val) \triangleq \val - \frac{1-\valdist(\val)}{\dens(\val)}$ is non-decreasing in $\val$.
\end{definition}

\begin{restatable}{proposition}{propRevRegularBudget}
\label{thm: private budget regular value}
A private-budget agent is 3-resemblant if her value and budget are independently distributed, and her marginal value distribution is regular.
\end{restatable}

The formal proof of the proposition can be found in \Cref{apx:rev close private budget}.
Moreover, in \cref{apx:rev close public budget}, we show that the approximation factor can be improved if the distribution over budgets is a point mass distribution, i.e., the budget of the agent is public.
In particular, an agent with a public budget and a regular value distribution is 1-resemblant.
That is, posting a per-unit price is revenue optimal for such a non-linear utility agent.
Hence, the value of detail for a public-budget agent is one: a seller who knows the demand curve gains \emph{nothing} from learning whether the agent's demand is limited by her budget or by her value.

\subsubsection{Correlated Budgets and Values}
\label{sec:revenue correlated}
The independence assumption in \cref{thm: private budget regular value} is essential: when budgets and values are correlated, the value of detail can be unbounded. If high-value agents have tighter budgets, the seller can attach to each option a small probability of charging the budget of the type for whom the option is designed. Higher-value agents cannot afford the options designed for lower-value agents, so the budget constraint screens agents, and the seller can charge each type nearly her value.

\begin{example}[Correlated budgets and values]
\label{example:correlated_main}
Fix $\kappa\geq 2$ and consider a single agent whose value $\val$ is drawn from $[1,\kappa]$ with density $\frac{\kappa}{\kappa-1}\frac{1}{\val^2}$ and whose budget is $\wealth=2\kappa-\val$, so that her value and budget are perfectly negatively correlated. Posting a per-unit price yields revenue at most $\frac{\kappa}{\kappa-1}\leq 2$. In contrast, the mechanism that charges $\val-2\epsilon$ with probability $1-\frac{\epsilon}{\kappa}$ and $\wealth$ with probability $\frac{\epsilon}{\kappa}$, for sufficiently small $\epsilon>0$, is incentive compatible and yields revenue at least $(1-\frac{2\epsilon}{\kappa})\frac{\kappa}{\kappa-1}\ln\kappa$. Hence, the value of detail is at least $(1-\frac{2\epsilon}{\kappa})\ln\kappa$, which is unbounded as $\kappa$ grows.
\end{example}

The mechanism relies on the joint distribution of values and budgets, which cannot be recovered from demand curves. Together with \cref{thm: private budget regular value}, the example shows that the value of detail for budget-constrained agents depends on how budgets relate to values: learning the details raises revenue by at most a factor of $3$ when budgets are independent of values, but it can raise revenue without bound when they are correlated, although only through extreme price discrimination that is arguably impractical.

\section{Conclusions and Extensions}
\label{sec:conclude}
Under non-linear utilities, agents cannot be ranked by a single willingness to pay, as linear agents can. We propose demand-based auctions, in which agents report demand curves, face per-unit prices set by the other agents' reports, and are compared through quantiles drawn from their demand curves. These auctions are dominant-strategy incentive compatible and ex post feasible, and their revenue depends only on the expected demand curves. With a single item and i.i.d.\ agents, the revenue-maximizing demand-based auction is a second-price auction with a monopoly reserve, run on values under risk aversion and on randomized budget caps under budget constraints.

When the seller observes only demand curves, this auction is max-min optimal, whereas a standard second-price auction can lose almost all of this revenue when budgets bind. When the seller also knows the utility functions and type distributions in full detail, the auction loses at most a factor $\gamma\zeta$ relative to the optimal mechanism, which may offer arbitrary menus of lotteries. Here, $\gamma$ is its approximation factor for linear agents, and $\zeta$ bounds the value of these details. Under standard assumptions on value distributions, this value is bounded for risk aversion of any degree and for budgets independent of values, but it can be arbitrarily large when budgets and values are correlated. Thus, the correlation between values and budgets, rather than the presence of budget constraints, is the salient feature for the value of detail.

We next discuss a simpler implementation and several extensions of our framework.

\subsection{Sequential Posted Pricing}
\label{sub:spp}
Although reports are simple in the two leading utility models (\cref{sub:demand_reports}), demand-based auctions can be complex to run in general: each agent may need to report an entire demand curve, and the seller must compute every agent's threshold from all reports; for marginal revenue maximization, this requires comparing ironed marginal revenues across agents. A simpler per-unit pricing mechanism that is widely used in practice is sequential posted pricing. The seller approaches the agents one at a time in a given order and offers each of them a per-unit price until the supply runs out. An agent reports nothing; she only decides how much to buy at the price she is offered.

{\sloppy
\begin{definition}
\label{def:spp}
A \emph{sequential posted pricing mechanism}
is parameterized by
$(\{\hat o_i\}_{i\in \agents},\{\hat\quant_i\}_{i\in \agents})$
where $\{\hat o_i\}_{i\in \agents}$ denotes
an order of the agents and
$\{\hat \quant_i\}_{i\in \agents}$
denotes the cutoffs in quantile space.
Agent $i$ receives an item if her quantile is at most~$\hat{\quant}_i$ and the items are not sold out to previous agents, i.e., the threshold functions $\{Q_i\}_{i\in \agents}$
satisfy that
$Q_i(\quant_{-i}) = \hat\quant_i \cdot
\indicator{\sum_{j:\hat o_j<\hat o_i}\indicator{\quant_j\leq\hat\quant_j}<m}$.
\end{definition}
}
In the value space, the mechanism sequentially offers the market clearing price $\price_i^{\hat{\quant}_i}$ to each agent~$i$ in the order $\{\hat o_i\}_{i\in \agents}$, as long as items remain. An agent's demand may be an interior lottery, which is realized immediately after her purchase decision. Since sequential posted pricing is a per-unit pricing mechanism, its revenue is determined by the price-posting revenue curves (\cref{thm:same_payoff}), and we denote by $\seqquantpayoff(\{\cumpayoff_i\}_{i\in \agents})$ the optimal revenue among sequential posted pricing mechanisms.

For linear agents, \citet{yan-11} shows that sequential posted pricing obtains at least a $1-\frac{1}{\sqrt{2\pi m}}$ fraction of the ex ante relaxation $\exanterelax(\{\payoffHull_i\}_{i\in \agents})$, which bounds the revenue of marginal revenue maximization from above. Because all three quantities depend only on the price-posting revenue curves,
\begin{align*}
\exanterelax(\{\payoffHull_i\}_{i\in \agents}) \geq \MRM(\{\cumpayoff_i\}_{i\in \agents}) \geq \seqquantpayoff(\{\cumpayoff_i\}_{i\in \agents}) \geq \left(1-\tfrac{1}{\sqrt{2\pi m}}\right)\cdot \exanterelax(\{\payoffHull_i\}_{i\in \agents})
\end{align*}
under every utility model, so sequential posted pricing is max-min optimal up to this factor when the seller observes only demand curves (cf.\ \cref{thm:robust_optimality}). When she knows the details, \cref{thm:meta thm expected utility} yields the same bound as for marginal revenue maximization (\cref{cor:mrm}).

\begin{corollary}
\label{cor:spp}
For selling $m$ units of identical items to any set of non-linear
agents with price-posting revenue curves
$\{\cumpayoff_i\}_{i\in \agents}$
whose ironed curves $\{\payoffHull_i\}_{i\in \agents}$ are $\zeta$-resemblant to their optimal revenue curves $\{\payoffcurve_i\}_{i\in \agents}$,
sequential posted pricing is a
$\sfrac{\zeta}{(1-\frac{1}{\sqrt{2\pi m}})}$-approximation to the ex ante relaxation for non-linear agents,
i.e., $\seqquantpayoff(\{\cumpayoff_i\}_{i\in \agents}) \geq
\frac{1}{\zeta} \cdot (1-\frac{1}{\sqrt{2\pi m}})
\cdot \exanterelax(\{\payoffcurve_i\}_{i\in \agents})$.
\end{corollary}

The guarantees of \cref{sec:guarantees} therefore extend to this simpler implementation. As for marginal revenue maximization, the factor $1-\frac{1}{\sqrt{2\pi m}}$ approaches one as the number of units $m$ grows, so in large markets, the seller can post prices sequentially with little loss relative to the revenue-maximizing demand-based auction, even when agents have non-linear utilities. Marginal revenue maximization nevertheless yields weakly higher revenue for every $m$, and in small markets, running the demand-based auction instead of posting prices sequentially can yield a non-negligible gain. The numerical illustration in \cref{sec:example} shows that the loss from sequential posted pricing is small already in markets of moderate size.

\subsection{Numerical Illustration}
\label{sec:example}

As with the value of detail in \cref{sec:revenue close}, the guarantees in \cref{sec:guarantees,sub:spp} concern the maximum possible loss across all environments in a class. They are driven by the least favorable environments, and the constants we derive are not necessarily tight, so they may considerably overstate the loss in typical environments.
In this subsection, we illustrate this point numerically by selling $m$ units of identical items to $n\geq m$ i.i.d.\ agents who are either risk averse or budget constrained. Since the optimal revenue is hard to compute for such agents, we compare the revenue of sequential posted pricing with the ex ante relaxation, an upper bound on the optimal revenue. The reported gaps therefore overstate the actual loss of sequential posted pricing, and since marginal revenue maximization earns at least as much, they also bound its loss. The details of the numerical computations, including the posted prices and the upper bound, are provided in \cref{apx:example_detail}.

\paragraph{Risk-Averse Utility}
We consider agents with risk-averse utilities.
Recall that a risk-averse agent's type is a pair $\type=(\val,\risk)$.
In particular, we focus on a simple example
where each agent has a constant absolute risk aversion (CARA) utility function,
i.e.,
\begin{align*}
\risk(z) = \frac{1}{a}\rbr{1-\exp\rbr{-az}}
\end{align*}
for risk parameter $a>0$, where $\exp(\cdot)$ is the exponential function,
and a private value drawn from a uniform distribution on $[0,1]$.

\begin{figure}[t]
\centering
\subfloat[risk-averse agents]{
\begin{tikzpicture}[scale = 0.53]

\draw (0, 0) -- (10.2, 0);
\draw (0, 0) -- (0, 5.2);

\begin{scope}[very thick]
\draw ( 1 , 2.5 ) --
( 2 , 3.75 ) --
( 3 , 4.166666666666667 ) --
( 4 , 4.375 ) --
( 5 , 4.500000000000001 ) --
( 6 , 4.583333333333332 ) --
( 7 , 4.642857142857143 ) --
( 8 , 4.6875 ) --
( 9 , 4.722222222222222 ) --
( 10 , 4.75 );
\end{scope}

\draw ( 1 , 1.25 ) --
( 2 , 2.1923828125 ) --
( 3 , 2.857095433587148 ) --
( 4 , 3.1809436354024 ) --
( 5 , 3.3727516719128388 ) --
( 6 , 3.499624544959603 ) --
( 7 , 3.5897709771934303 ) --
( 8 , 3.657124379028153 ) --
( 9 , 3.709360289813776 ) --
( 10 , 3.751055366374377 );

\draw [dashed] ( 1 , 1.25 ) --
( 2 , 2.343286640475526 ) --
( 3 , 3.0917220601082693 ) --
( 4 , 3.461598998423752 ) --
( 5 , 3.6822393126117587 ) --
( 6 , 3.828817015045652 ) --
( 7 , 3.9332684417703887 ) --
( 8 , 4.011474041027612 ) --
( 9 , 4.072222694718358 ) --
( 10 , 4.120772938674861 );

\draw (1, -0.5) node {$1$};
\draw [dotted] (1, 0) -- (1, 2.5);

\draw (0, -0.5) node {$0$};
\draw (10, 0) -- (10, 0.2);
\draw (10, -0.5) node {$10$};
\draw (11, 0) node {$\frac{n}{m}$};

\draw (1, 6) node {$\frac{1}{m}\cdot$ revenue};

\draw (0, 5) -- (0.2, 5);
\draw (-0.3, 5) node {$1$};

\end{tikzpicture}

\label{fig:rev_risk}
}
\subfloat[private-budget agents]{
\begin{tikzpicture}[scale = 0.53]

\draw (0, 0) -- (10.2, 0);
\draw (0, 0) -- (0, 5.2);

\begin{scope}[very thick]
\draw ( 1.0 , 0.9708338949998667 ) --
( 2.0 , 1.9416677899997448 ) --
( 3.0 , 2.7269737461427677 ) --
( 4.0 , 3.1819317908752645 ) --
( 5.0 , 3.480176984106179 ) --
( 6.0 , 3.6922058035500536 ) --
( 7.0 , 3.8513148812073945 ) --
( 8.0 , 3.975394828880128 ) --
( 9.0 , 4.075024080948824 ) --
( 10.0 , 4.156852276724797 );
\end{scope}

\draw ( 1 , 0.9622504486493767 ) --
( 2 , 1.7485292424851289 ) --
( 3 , 2.3252737992263075 ) --
( 4 , 2.6888499424515953 ) --
( 5 , 2.9249102223408014 ) --
( 6 , 3.091865304686732 ) --
( 7 , 3.2167486713908557 ) --
( 8 , 3.3139532805586063 ) --
( 9 , 3.391904686865634 ) --
( 10 , 3.4558873484311334 );

\draw [dashed] ( 1 , 0.9622504486493776 ) --
( 2 , 1.8657886763134734 ) --
( 3 , 2.5162268702497146 ) --
( 4 , 2.9260877697145204 ) --
( 5 , 3.1933034074978384 ) --
( 6 , 3.382701868365859 ) --
( 7 , 3.5245524337545304 ) --
( 8 , 3.6350520738022754 ) --
( 9 , 3.7237124908320096 ) --
( 10 , 3.7965120942185813 );

\draw (0, -0.5) node {$0$};
\draw (10, 0) -- (10, 0.2);
\draw (10, -0.5) node {$10$};
\draw (11, 0) node {$\frac{n}{m}$};

\draw (1, 6) node {$\frac{1}{m}\cdot$ revenue};

\draw (0, 5) -- (0.2, 5);
\draw (-0.3, 5) node {$1$};

\draw (1, -0.5) node {$1$};
\draw [dotted] (1, 0) -- (1, 0.96225);

\end{tikzpicture}

\label{fig:rev_budget}
}
\caption{The revenue (normalized by the number of items $m$) as a function of $\frac{n}{m}$ for both risk-averse agents and private-budget agents.
The expected revenue for sequential posted pricing is computed when $m=5$ (thin solid line) and $m=20$ (thin dashed line).
An upper bound on the optimal revenue through ex ante relaxation is illustrated by the thick solid line.
Note that this normalized upper bound does not depend on the value of $m$.
}
\label{fig:rev_example}
\end{figure}

\paragraph{Private-Budget Utility}
We now consider agents with private-budget utilities
where the values are independent of the budgets.
This model is a generalization of the single-item setting studied in \citet{CG-00} for a single agent
and in \citet{PV-14} for multiple agents.
In both scenarios, the optimal mechanisms are complicated.
For numerical analysis,
we consider budget-constrained agents whose value and budget are drawn independently from the uniform distribution on $[0,1]$.

\paragraph{Numerical Evaluations}
As illustrated in \Cref{fig:rev_example},
the gap between the expected revenue from sequential posted pricing
and the upper bound of the optimal revenue is small,
especially when both $n$ and $m$ are large.
For example, the multiplicative revenue gap when $n=200$ and $m=20$ is less than $1.153$ for risk-averse agents and less than $1.095$ for private-budget agents.

\subsection{General Objectives}
\label{sub:general_objectives}
Although the main text focuses on revenue maximization, our framework applies to the general payoff functions of \cref{sec:prelim}, which are additive and separable across agents, once price-posting revenue curves are replaced by price-posting payoff curves. The implementation of per-unit pricing mechanisms (\cref{thm:same_payoff}) carries over directly, and the max-min optimality result (\cref{thm:robust_optimality}) holds whenever these curves admit a linear agents analog. The reduction framework (\cref{thm:meta thm expected utility}) also applies, so approximation guarantees for a given objective follow from the corresponding resemblance bounds. For welfare maximization, an agent whose budget is drawn independently of her value is $2$-resemblant, and risk-averse agents require no approximation: they behave as linear agents when facing per-unit prices, so the VCG mechanism remains welfare optimal for them. Finally, resemblance for two objectives implies resemblance for any convex combination of them (see \cref{apx:general objectives} for details).

\subsection{Heterogeneous Utility Models}
\label{sub:heterogeneous}
Our framework also applies when agents have heterogeneous utility models, for example, when some agents have private budget constraints and others are risk averse. Resemblance is monotone (\cref{lem:resemblant is downward imply}), so if each agent $i \in N$ is $\zeta_i$-resemblant, every agent is $\max_i\{\zeta_i\}$-resemblant, and sequential posted pricing is a $\sfrac{\max_i\{\zeta_i\}}{(1-\frac{1}{\sqrt{2\pi m}})}$-approximation to the ex ante relaxation. The proof of \cref{lem:resemblant is downward imply} is immediate from \cref{def:zeta closeness}: since $\sfrac{1}{\zeta'}\leq\sfrac{1}{\zeta}$, any quantile $\quant\leq\quant\primed$ with $\payoffHull(\quant)\geq\sfrac{1}{\zeta}\cdot\payoffcurve(\quant\primed)$ also satisfies $\payoffHull(\quant)\geq\sfrac{1}{\zeta'}\cdot\payoffcurve(\quant\primed)$.
\begin{lemma}\label{lem:resemblant is downward imply}
For any $\zeta' \geq \zeta \geq 1$,
$\zeta$-resemblant implies $\zeta'$-resemblant.
\end{lemma}

\subsection{Dominant-Strategy Implementation}
All mechanisms in this paper are dominant-strategy incentive compatible. Unlike for linear agents, for whom any Bayesian incentive compatible mechanism can be implemented in dominant strategies in single-item auctions \citep{GGKMS-13}, restricting attention to dominant-strategy incentive compatible mechanisms is not without loss for non-linear agents \citep[e.g.,][]{FH-18,FLLT-18}. Our results thus have implications for the design of strategically simple mechanisms \citep[e.g.,][]{CE-07,BL-19}: dominant-strategy incentive compatible mechanisms are robustly optimal, since marginal revenue maximization is max-min optimal when the seller observes only the expected demand curves (\cref{thm:robust_optimality}), and for a broad family of non-linear agents, they are approximately optimal for any convex combination of welfare and revenue as the objective function.

\bibliographystyle{apalike}
\bibliography{auctions}

\newpage
\appendix
\section{Omitted Proofs}
\label{apx:rev close}

\subsection{Demand-Based Auctions}
\label{apx:demand_auction}
\begin{proof}[Proof of \cref{prop:demand_auction}]
Fix agent $i$ and the reports of the other agents. The threshold $\exquanti=Q_i(\quant_{-i})$ does not depend on agent $i$'s report. Conditional on $\exquanti$, a report $\hat\demand_i$ gives agent $i$ an item with probability $z=\hat\demand_i(\price_i^{\exquanti})$ at the deterministic payment $z\cdot\price_i^{\exquanti}$, which is the option $(z,z\cdot\price_i^{\exquanti})$ of the per-unit price $\price_i^{\exquanti}$ (\cref{def:per-unit}). The truthful report selects her demand $\demand_i^{\typei}(\price_i^{\exquanti})$, which is optimal in this menu for every realization of $\exquanti$, so truthful reporting is a dominant strategy. Since the option $z=0$ is always available, the auction is IIR. It is ex post feasible because items are allocated by the thresholds of a feasible mechanism for linear agents. Under truthful reporting, the quantile of agent $i$ is drawn from $H_i^{\typei}$, so the outcome coincides with that of \cref{alg:meta theorem}, and the revenue follows from \cref{thm:same_payoff}. Finally, the seller uses the expected demand curves to compute the market clearing prices $\price_i^{\quant}$ and the thresholds $Q_i$. For marginal revenue maximization, the thresholds depend on the price-posting revenue curves $\cumpayoff_i(\quant)=\quant\cdot\price_i^{\quant}$, which are also determined by the expected demand curves.
\end{proof}

\subsection{Robust Optimality}
\label{apx:robust}
\begin{proof}[Proof of \cref{thm:robust_optimality}]
Note that in the robust environments, the seller does not know the detailed utility model to compute the demands based on the reported types.
The demand-based auction (\cref{def:demand_auction}) does not require such knowledge: it elicits the demand curves from the agents and uses only the expected demand curves. By \cref{prop:demand_auction}, truthful reporting is a dominant strategy, and the expected revenue of the marginal revenue maximization mechanism~is
\begin{align*}
\MRM(\lbr{\cumprice_i}_{i\in \agents}) =  \max_{\text{per-unit pricing mech. } \mech\primed} \mech\primed(\lbr{\cumprice_i}_{i\in \agents})
\end{align*}
regardless of the underlying utility model.
Moreover, given any non-increasing demand curve~$\aggregateDemand_i$,
there exists a linear utility model $(\typespace^*_i, \typedist^*_i, \util^*_i)$
with $\typespace^*_i =\reals$,
$\util^*_i(\type,\alloc,\price) = \type\cdot\alloc-\price$
such that $(\typespace^*_i, \typedist^*_i, \util^*_i)\in \feasibleUF(\aggregateDemand_i)$.
For agents with linear utilities, marginal revenue maximization is revenue optimal even when the seller knows the detailed utility model \citep{BR-89}.
By the revelation principle, the revenue of any Bayes--Nash equilibrium of any indirect mechanism under the linear utility models is attained by a BIC and IIR revelation mechanism for these models, so the following inequality holds for every equilibrium selection.
Therefore,
for any mechanism $\mech$, we have
\begin{align*}
&\min_{(\typespace_i, \typedist_i, \util_i) \in \feasibleUF(\aggregateDemand_i), \forall i\in \agents}
\mech(\lbr{(\typespace_i, \typedist_i, \util_i)}_{i\in \agents})
\leq \mech(\lbr{(\typespace^*_i, \typedist^*_i, \util^*_i)}_{i\in \agents})\\
&\leq \max_{\text{per-unit pricing mech. } \mech\primed} \mech\primed(\lbr{\cumprice_i}_{i\in \agents})
= \MRM(\lbr{\cumprice_i}_{i\in \agents}).
\end{align*}
That is, the marginal revenue maximization mechanism is robustly optimal for non-linear agents.
\end{proof}

\subsection{Risk Aversion}
\label{apx:proof rev risk}
We first show that given per-unit prices, a risk-averse agent behaves in the same way as a linear agent.
\begin{proof}[Proof of \cref{clm:risk_demand}]
Given any per-unit price $\price$, agent's expected utility with demand $z$ is
\begin{align*}
z\cdot \risk(\val-z\cdot\price)
+ (1-z) \cdot \risk(-z\cdot\price)
\leq \risk(z\cdot(\val-\price))
\end{align*}
by Jensen's inequality since function $\risk$ is concave,
and the equality holds when $z\in\{0, 1\}$.
Moreover, since $\risk$ is monotonically increasing and $z\cdot(\val-\price)$ is maximized when $z$ is either $0$ or~$1$,
the expected utility is maximized at allocation $z\in\{0, 1\}$.
\end{proof}

Next we prove the bound on resemblance for risk-averse agents.
\riskaversebound*

Our proof relies on the following three auxiliary lemmas.

\begin{lemma}
\label{lem:risk averse revenue upper bounded by allocated value}
    For every risk-averse agent and every mechanism that is BIC and IIR, the expected revenue is at most the expected allocated value.
\end{lemma}
\begin{proof}
    Let $\alloc$ and $\price$ be the (randomized) allocation and payment rule of the BIC, IIR mechanism. For an agent with risk-averse type $\type = (\val, \risk)$, her expected utility is $\expect[]{\risk(\val\alloc(\type) - \price(\type))}$.
    Since the mechanism is individually rational and $\risk$ is concave,
    \begin{align*}
        \risk\left(\expect[]{\val\alloc(\type) - \price(\type)}\right)
        \geq  \expect[]{\risk(\val\alloc(\type) - \price(\type))}
        \geq  0
    \end{align*}
    Combining with the assumption that  $\risk$ is strictly increasing
    with $\risk(0) = 0$, we obtain
        $\expect{\val\alloc(\type)} \geq \expect{\price(\type)}$ as desired.
\end{proof}

\begin{lemma}
\label{lem:conditional dist}
    For any MHR distribution $\valdist$ and any threshold $\hat\val\in \reals_+$, the conditional distribution $\valdist_{\geq \hat\val}$ (conditioning on random variable $\val\sim\valdist$ is at least $\hat\val$) is also MHR.
\end{lemma}
\begin{proof}
    It suffices to verify the monotonicity of the hazard rate function $h_{\geq \hat\val}(\cdot)$ for conditional distribution $\valdist_{\geq \hat\val}$.
    Let $\dens$, $\dens_{\geq\hat\val}$ be the density function of distribution $\valdist$ and $\valdist_{\geq\hat\val}$, respectively.
    Note that
    \begin{align*}
        h_{\geq \hat\val}(\val) =
        \frac{\dens_{\geq \hat\val}(\val)}{1 - \valdist_{\geq\hat\val}(\val)}
        =
        \frac{\frac{\dens(\val)}{(1-\valdist(\hat\val))}}
        {\frac{1 - \valdist(\val)}{(1-\valdist(\hat\val))}}
        =\frac{{\dens(\val)}}
        {{1 - \valdist(\val)}}
        =
        h(\val)
    \end{align*}
    where $h(\cdot)$ is the hazard rate function for distribution $\valdist$, the first and last equalities hold by definition, and the second equality holds due to the definition of conditional distribution. Finally, invoking the monotonicity of hazard rate function $h(\cdot)$ finishes the proof.
\end{proof}

\begin{lemma}[\citealp{HR-09}]
\label{lem:mhr revenue approx welfare}
    For any MHR value distribution $\valdist$, it is guaranteed that $\max_{\price\in\reals_+} \price(1-\valdist(\price)) \geq
    \frac{1}{e} \cdot \expect[\val\sim \valdist]{\val}$.
\end{lemma}
As a sanity check, \Cref{lem:mhr revenue approx welfare} suggests that for a linear agent with value drawn from MHR distribution $\valdist$, the optimal revenue $\max_{\price\in\reals_+} \price(1-\valdist(\price))$ is at least an $\frac{1}{e}$-fraction of the expected full surplus $\expect[\val\sim \valdist]{\val}$.

Now we are ready to prove \Cref{thm:risk-averse bound}.

\newcommand{\quantHat}{\hat\quant}

\begin{proof}[Proof of \Cref{thm:risk-averse bound}]
    Fix an arbitrary $\quantHat \in [0, 1]$.
    Let $\valdist$ be the marginal value distribution.
    To simplify the presentation, we assume $\valdist$ is a continuous distribution. The analysis for general distribution can be extended straightforwardly.
    With slight abuse of notation, we denote
    $\val(\quant) \triangleq \valdist^{-1}(1-\quant)$ for each quantile $\quant\in[0, 1]$.
    Invoking \Cref{lem:risk averse revenue upper bounded by allocated value},
    the expected revenue of the $\hat\quant$ ex ante revenue-optimal mechanism
    is
    at most $\expect[\quant\sim U\text{[}0, 1\text{]}]{\val(\quant)\cdot \indicator{\quant \leq \hat\quant}} =
    \hat\quant\cdot \expect[\quant'\sim U\text{[}0, \hat\quant\text{]}]{\val(\quant')}
    $.
    Note that the random variable $\val(\quant')$ where $\quant'\sim U[0, \hat\quant]$ can be interpreted as the random variable drawn from conditional distribution $\valdist_{\geq \val(\hat\quant)}$.
    Invoking \Cref{lem:conditional dist,lem:mhr revenue approx welfare}, it ensures that there exists a price $\price \geq \val(\hat\quant)$ such that its expected revenue is at least
    $\frac{1}{e}\cdot\expect[\quant\sim U\text{[}0, 1\text{]}]{\val(\quant)\cdot \indicator{\quant \leq \hat\quant}}$.
    Since $\price\geq \val(\hat\quant)$,
    the expected allocation probability is at most $\hat\quant$.
    Therefore, $\max_{\quant\leq \hat\quant}\cumpayoff(\quant) \geq \frac{1}{e}\cdot \payoffcurve(\hat\quant)$ as desired.
\end{proof}

\subsection{Private Budget}
\label{apx:rev close private budget}

\propRevRegularBudget*
Fix an arbitrary ex ante constraint $\quant$,
denote $\EX$ as the $\quant$
ex ante revenue-optimal mechanism,
and $\Rev{\EX}$ as its revenue.
We decompose $\EX$ into two mechanisms $\EXS$ and $\EXL$
according to the market clearing price $\marketclearing$.
Intuitively, the per-unit prices in $\EXS$ for all types are at most the market clearing price
and the per-unit prices in $\EXL$ for all types are
larger than the market clearing price.
Specifically, for each budget level $\wealth$, any incentive compatible mechanism offers the types with budget $\wealth$ a convex and non-decreasing allocation-payment function $\APF_\wealth$, which maps an allocation to its payment (\cref{lem:pricing function characterization}). Let $\alloc^*_\wealth$ be the largest allocation at which the marginal price of $\APF_\wealth$ is at most $\marketclearing$. Mechanism $\EXS$ offers $\APF_\wealth$ restricted to allocations at most $\alloc^*_\wealth$, and $\EXL$ offers the remainder, $\APFL_\wealth(\alloc)=\APF_\wealth(\alloc^*_\wealth+\alloc)-\APF_\wealth(\alloc^*_\wealth)$. Both ignore the incentive constraints on budget reports, and we refer to this relaxation as the \emph{random-public-budget} utility model; the market clearing price $\marketclearing$ is the same under both models. A type whose allocation in $\EX$ is $\alloc\geq\alloc^*_\wealth$ purchases $\alloc^*_\wealth$ in $\EXS$, which she can afford since $\APF_\wealth(\alloc^*_\wealth)\leq\APF_\wealth(\alloc)\leq\wealth$, and at least $\alloc-\alloc^*_\wealth$ in $\EXL$, so her payments in the two mechanisms sum to at least her payment $\APF_\wealth(\alloc)$ in $\EX$; a type with $\alloc<\alloc^*_\wealth$ makes the same purchase in $\EXS$ as in $\EX$. Hence $\Rev{\EX}\leq\Rev{\EXS}+\Rev{\EXL}$. Moreover, since $\APFL_\wealth(\alloc)\geq\APF_\wealth(\alloc)$ by convexity, each type purchases weakly less in $\EXS$ and in $\EXL$ than in $\EX$.
The formal details of the decomposition are specified in \Cref{apx:welfare budget},
and we will bound the revenue from those two mechanisms separately.
\begin{lemma}
\label{lem:EXS}
For a single agent with random-public-budget utility, independently
distributed value and budget, and any ex ante constraint $\quant$; the
revenue of $\EXS$ is at most the revenue from posting the market
clearing price, i.e., $\cumprice(\quant) \geq \Rev{\EXS}$.
\end{lemma}
\begin{proof}
The ex ante allocation of $\EXS$
is at most the ex ante allocation of $\EX$,
i.e.,~$\quant$.
Combining with the fact that the per-unit prices
in $\EXS$ for all types are weakly lower
than the market clearing price,
its revenue is at most the revenue
of posting the market clearing price.
\end{proof}

For the revenue bound of $\EXL$, we consider two different cases:
(1) the market clearing price is larger than the monopoly reserve $\myersonReserve\in\argmax_{\price}\price\,(1-\valdist(\price))$ of the value distribution;
and (2) the market clearing price is smaller than the monopoly reserve.
The proofs of \cref{lem:market clearing price large,lem:posting randomized price for EXL} are provided at the end of this section.

\begin{lemma}
\label{lem:market clearing price large}
For a single private-budget agent with independently distributed value
and budget and regular value distribution, if the market clearing
price $\marketclearing = \sfrac{\cumprice(\quant)}{\quant}$ is larger
than the monopoly reserve,
i.e.,
$\marketclearing = \cumprice(\quant)/\quant
\geq \myersonReserve$,
the revenue of posting the market clearing price
is at least the revenue of $\EXL$,
i.e., $\cumprice(\quant) \geq \Rev{\EXL}$.
\end{lemma}
Intuitively, \cref{lem:market clearing price large} is analyzed separately for each realized budget under two cases.
If the budget is lower than the market clearing price, then the budget binds, and posting the market clearing price extracts the entire budget, which is at least the revenue of $\EXL$.
If the budget is higher than the market clearing price, then the budget does not bind,
and the revenue comparison is analogous to the setting without budget constraints.
Since the revenue curve is concave and the market clearing price is larger than the monopoly reserve,
higher price in $\EXL$ leads to lower expected revenue.

\begin{figure}[t]
\begin{flushleft}
\hspace{-5pt}
\begin{minipage}[t]{0.48\textwidth}
\centering
\begin{tikzpicture}[scale = 0.4]

\draw [name path = A, white, dotted] (9.1, 0) -- (9.1, 2.654);
\draw [name path = D, white, dotted] (0, 0) -- (9.1, 2.654);

\draw [name path = C, dashed, thick] (0, 0) -- (12, 5.3);
\draw [dotted, thick] (0, 0) -- (12, 3.5);

\tikzfillbetween[of=A and D]{gray!40!white};
\fill[color=gray!40!white] (9.1, 2.654) -- (12, 2.654)
-- (12, 0) -- (9.1, 0);

\draw [gray, dotted] (9.1, 0) -- (9.1, 4.05);

\fill[pattern=north west lines, pattern color=gray!20!white] (0, 0) -- (0, 2.654)
-- (12, 2.654) -- (12, 0);

\draw (-0.2,0) -- (12.5, 0);
\draw (0, -0.2) -- (0, 6);

\begin{scope}[very thick]

\draw [name path global = A](0, 0) .. controls
(2, 5.5) .. (6, 4.714285714285714);

\draw [name path global = F]
(6, 4.714285714285714) .. controls
(11, 3.75) .. (12, 0);

\end{scope}

\draw (0, -0.8) node {$0$};
\draw (12, -0.8) node {$1$};

\draw (10.7, 5.3) node {$\marketclearing$};
\draw (12, 4) node {$\wealth$};

\draw [gray, dotted] (0, 2.654) -- (12, 2.654);
\draw [gray, dotted] (12, 0) -- (12, 2.654);
\draw (-2.35, 2.654) node {$\RevPPw$};

\draw (-0.2, 2.654) -- (0.2, 2.654);

\end{tikzpicture}

\end{minipage}
\begin{minipage}[t]{0.48\textwidth}
\centering
\begin{tikzpicture}[scale = 0.4]

\draw [name path = E, white] (6, 4.714285714285714) -- (0, 0);

\draw [name path global = B, white]
(6, 4.714285714285714) .. controls
(11, 3.75) .. (12, 0);

\fill[color=gray!40!white] (0, 0) -- (6, 4.714285714285714)
-- (12, 0);
\draw[color=gray!40!white]
(6, 4.714285714285714) -- (12,0);
\fill[color=gray!40!white] (9, 4.05) -- (12, 4.05)
-- (12, 0) -- (9, 0);

\tikzfillbetween[of=B and E]{gray!40!white};

\fill[pattern=north west lines, pattern color=gray!20!white] (0, 0) -- (0, 4.714285714285714)
-- (12, 4.714285714285714) -- (12, 0);

\draw (-0.2,0) -- (12.5, 0);
\draw (0, -0.2) -- (0, 5.2);

\begin{scope}[very thick]

\draw [name path global = A](0, 0) .. controls
(2, 5.5) .. (6, 4.714285714285714);

\draw [name path global = F]
(6, 4.714285714285714) .. controls
(11, 3.75) .. (12, 0);

\end{scope}

\draw [name path = C, dashed, thick] (0, 0) -- (12, 5.3);
\draw [name path = D, dotted, thick] (0, 0) -- (7, 5.5);

\draw (0, -0.8) node {$0$};
\draw (12, -0.8) node {$1$};

\draw (10.7, 5.3) node {$\marketclearing$};
\draw (7.5, 5.5) node {$\wealth$};

\draw [gray, dotted] (0, 4.714285714285714) -- (12, 4.714285714285714);
\draw [gray, dotted] (12, 0) -- (12, 4.714285714285714);
\draw (-0, 6) node
{$\Rev[\wealth]{\OPT_\wealth}$};

\draw [gray, dotted] (9, 4.05) -- (12, 4.05);
\draw [gray, dotted] (12, 0) -- (12, 4.05);

\draw (-0.2, 4.714285714285714) -- (0.2, 4.714285714285714);

\end{tikzpicture}

\end{minipage}
\end{flushleft}
\caption{\label{f:revenue of sample} In the geometric proof of
  \Cref{lem:posting randomized price for EXL}, the upper bound on the
  expected revenue of $\EXL$ ($\RevPPw$ and
  $\Rev[\wealth]{\OPT_\wealth}$ on the left and right, respectively)
  is the area of the light gray striped rectangle and the revenue from
  posting random price $\randomprice$ is the area of the dark gray
  region.  By geometry, the latter is at least half of the former.
  The black curve is the price-posting revenue curve with no budget
  constraint $\cumpricelinear$.  The figure on the left depicts the
  small-budget case (i.e., $\wealth < \marketclearing$), and the
  figure on the right depicts the large-budget case (i.e., $\wealth
  \geq \marketclearing$).  }
\end{figure}

\begin{lemma}
\label{lem:posting randomized price for EXL}
For a single private-budget agent with independently distributed value
and budget and regular value distribution, if the market clearing
price $\marketclearing = \sfrac{\cumprice(\quant)}{\quant}$ is smaller
than the monopoly reserve, there exists $\quant\primed \leq \quant$
such that the market clearing revenue from $\quant\primed$ is a
2-approximation to the revenue from $\EXL$, i.e., $2
\cumprice(\quant\primed) \geq \Rev{\EXL} $.
\end{lemma}
In fact, we show there is a distribution over interim feasible random pricing that attains at least half of the optimal ex ante revenue.
Of course, the optimal deterministic price that is at least $\marketclearing$ is only better than the random price,
and hence the lemma is shown.
Specifically, consider posting a random price
$\randomprice = \max\{\marketclearing, \randomprice_0\}$ with
$\randomprice_0$ drawn identically to the agent's value distribution.
As shown in \Cref{f:revenue of sample}, the revenue from the random pricing $\randomprice$ corresponds to the dark gray area,
while an upper bound on the revenue of $\EXL$ corresponds to the light gray striped area.
By the concavity of the revenue curve, the geometry of the figure directly implies the 2-approximation on expected revenue.

\begin{proof}
[Proof of \Cref{thm: private budget regular value}]
Fix any ex ante constraint $\quant$.
If the market clearing price $\marketclearing =
\sfrac{\cumprice(\quant)}{\quant}$
is at least the monopoly reserve,
\Cref{lem:EXS} and \Cref{lem:market clearing price large}
imply that
$\Rev{\EXS} \leq \cumprice(\quant)$,
and
$\Rev{\EXL} \leq \cumprice(\quant)$,
thus,
$\cumprice(\quant)$
is a 2-approximation to
$\Rev{\EXS} + \Rev{\EXL} \geq \Rev{\EX}$,
i.e., $\revcurve(\quant)$.
If the market clearing price $\marketclearing$
is smaller than the monopoly reserve,
let
$\quant\primed
=\argmax_{\quant' \leq \quant}
\cumprice(\quant')$,
\Cref{lem:EXS} and
\Cref{lem:posting randomized price for EXL}
imply that
$\Rev{\EXS} \leq \cumprice(\quant) \leq
\cumprice(\quant\primed)$,
and
$\Rev{\EXL} \leq 2\cumprice(\quant\primed)$,
thus, $\cumprice(\quant\primed)$
is a 3-approximation to
$\revcurve(\quant)$.
Thus, the agent is 3-resemblant for ex ante optimization.
\end{proof}

\begin{proof}[Proof of \cref{lem:market clearing price large}]
In both $\EXL$ and
the mechanism that posts the market clearing price,
the types with value lower than
the market clearing price
$\marketclearing$ will purchase nothing,
so we only consider the types with value at least $\marketclearing$
in this proof.
Each budget level is considered separately.

For types with budget $\wealth \leq \marketclearing$,
by posting the market clearing price $\marketclearing$,
those types always pay their budgets $\wealth$,
which is at least the revenue
from those types in $\EXL$.

For types with budget $\wealth > \marketclearing$,
by posting the market clearing price $\marketclearing$,
those types always pay $\marketclearing$.
Since the budget constraints do
not bind for these types,
it is helpful to consider
the price-posting revenue curve without budget,
which we denote by $\cumpricelinear$.
The regularity of the value distribution
guarantees that
 $\cumpricelinear$ is concave.
The concavity of $\cumpricelinear$
implies that prices above $\marketclearing$, which is at least $\myersonReserve$,
extract lower revenue
than $\marketclearing$.
Since the per-unit prices in $\EXL$ for all types
are at least $\marketclearing$,
the concavity of $\cumpricelinear$
guarantees that the expected revenue of posting $\marketclearing$
for types with budget larger than
$\marketclearing$
is at least the expected revenue for those types in $\EXL$.
Combining these bounds above,
we have
$\cumprice(\quant) \geq \Rev{\EXL}$.
\end{proof}

\begin{proof}[Proof of \cref{lem:posting randomized price for EXL}]
Note that any price that is at least $\marketclearing$ is feasible for
the ex ante constraint $\quant$.  We consider posting a random price
$\randomprice = \max\{\marketclearing, \randomprice_0\}$ with
$\randomprice_0$ drawn identically to the agent's value distribution.
Fixing the budget of the agent $\wealth$, consider the following
geometric argument \citep[cf.][]{DRY-10}.  For both sides of
\Cref{f:revenue of sample}, the area of the light gray striped
rectangle upper bounds the revenue of $\EXL$ and the area of the dark
gray region is the expected revenue from posting random price
$\randomprice$.  Consequently, concavity of the price-posting revenue
curve with no budget constraint $\cumpricelinear$ (by regularity of
the value distribution) implies that a triangle with half the area of
the light gray rectangle is contained within the dark gray region and,
thus, the random price is a 2-approximation.  As the random price does
not depend on the budget $\wealth$, the same bound holds when
$\wealth$ is random.  Of course, the optimal deterministic price that
is at least $\marketclearing$ is only better than the random price and
the lemma is shown.  The remainder of this proof verifies that the
geometry of the regions described above is correct.

The left side of \Cref{f:revenue of sample} depicts the fixed budgets
$\wealth$ that are at most $\marketclearing$.  The area of the light
gray striped rectangle upper bounds the revenue of $\EXL$ as follows.
Let $\Rev[\wealth]{\price}$ be the expected revenue from posting price
$\price$ to types with budget $\wealth$.  Under both $\EXL$ and the
market clearing price $\marketclearing$, types with value below the
market clearing price pay zero.  For the remaining types, in $\EXL$
they pay at most their budget and in market clearing they pay exactly
their budget.  Thus, $\Rev[\wealth]{\EXL} \leq
\Rev[\wealth]{\marketclearing} = \wealth \, (1-F(\marketclearing))$
where, recall, $1-F(\marketclearing)$ is the probability the agent's
value is at least the market clearing price $\marketclearing$.  Of
course, $\wealth \, (1-F(\marketclearing))$ is the height and area (its
width is 1) of the light gray striped region on the left side of
\Cref{f:revenue of sample}.

The right side of \Cref{f:revenue of sample} depicts the fixed budgets
$\wealth$ that are at least $\marketclearing$.  The area of the light
gray striped rectangle upper bounds the revenue of $\EXL$ as follows.
Let $\OPT_\wealth$ be the optimal mechanism to types with budget
$\wealth$ without ex ante constraint and $\Rev[\wealth]{\OPT_\wealth}$
be its expected revenue from these types.  Clearly,
$\Rev[\wealth]{\EXL} \leq \Rev[\wealth]{\OPT_\wealth}$ as the latter
optimizes with relaxed constraints of the former.
\cite{LR-96} show that $\OPT_\wealth$ posts the
minimum between budget $\wealth$ and the monopoly reserve
$\myersonReserve$
when the agent has public budget and regular value distribution.
As the budget does not bind for this price, its
revenue is given by the price-posting revenue curve with no budget
constraint, i.e., $\Rev[\wealth]{\OPT_\wealth} = \cumpricelinear(1 -
F(\min\{\wealth, \myersonReserve\}))$.  Of course, this revenue is the
height and area (its width is 1) of the light gray striped region on
the right side of \Cref{f:revenue of sample}.

Next, we will show that the revenue of posting the random price
$\randomprice$ is the gray shaded areas illustrated in \Cref{f:revenue
  of sample} (in both cases).  A random price from the value
distribution, i.e.,  $\randomprice_0$, corresponds to a uniform random quantile
constraint, i.e., drawing uniformly from the horizontal axis.  Since
we truncate the lower end of the price distribution at the market
clearing price $\marketclearing$, the revenue from quantiles
at least $1-F(\marketclearing)$ equals the revenue from the market clearing price. For
any fixed $\wealth$, when $\randomprice \in [\marketclearing,
  \wealth]$, the budget does not bind and the revenue of posting price
$\randomprice$ is $\cumpricelinear(\randomquant)$ where
$\cumpricelinear$ is the price-posting revenue curve without budget;
and when $\randomprice > \wealth$, the revenue of posting
price~$\randomprice$ is~$\wealth \randomquant$.  Thus, the revenue
from a random price is given by the integral of the area under the
curve defined by $\randomquant \wealth$ when $\randomprice \geq
\wealth$, by $\cumpricelinear(\randomquant)$ when $\randomprice \in [\marketclearing, \wealth]$
and this interval exists, and by
$\min(\wealth,\marketclearing)\,(1-F(\marketclearing))$
when $\randomprice = \marketclearing$,
i.e., when $\randomprice_0 \leq \marketclearing$.
This area is the dark gray region.
\end{proof}

\subsection{Necessity of Assumptions}
\label{apx:necessity_rev_resemble_assuption}

The following example shows that the regularity assumption in \cref{thm: private budget regular value} cannot be dropped, even when budgets are independent of values.

\begin{example}
\label{example:post price is bad}
Consider a private-budget agent where the budget distribution is the discrete equal revenue distribution,
i.e., $g(i) = \frac{1}{\normalization\, i^2}$,
where $\normalization = \sfrac{\pi^2}{6}$.
Let the quantile function of the value distribution be $\quant(i) = \min\{1,\sfrac{1}{\ln i}\}$.
The optimal price-posting revenue is a constant.
Next consider the pricing function
$\APF(\alloc) = \frac{1}{1-\alloc}$.
From this pricing function,
the value $\val_i$ corresponding to payment $i$ is
$\val_i = i^2$.
Note that the revenue from this payment function is infinity, i.e.,
\begin{eqnarray*}
    \Rev{\APF}
    &\geq& \lim_{\budgetnum \to \infty}
    \sum_{i = 2}^\budgetnum
    \left( i \cdot \quant(\val_i) \cdot g(i) \right) \\
    &=& \frac{1}{2\normalization}
    \lim_{\budgetnum \to \infty}
    \sum_{i = 2}^\budgetnum
    \frac{1}{i \cdot \ln i} \\
    &\geq& \frac{1}{2\normalization}
    \lim_{\budgetnum \to \infty}
    \big(\ln \ln (\budgetnum+1) - \ln\ln 2\big) \to \infty.
\end{eqnarray*}
Therefore, the gap between price posting and the optimal mechanism is infinite.
\end{example}

\newpage
\begin{center}
\large
\MakeUppercase{Online Appendix}
\MakeUppercase{For Online Publication}
\end{center}
\section{General Objectives}
\label{apx:general objectives}
The main text focuses on revenue maximization. This appendix extends the framework to the general payoff functions introduced in \cref{sec:prelim}, which are additive and separable across agents. \Cref{apx:general framework} defines the payoff curves and shows that the implementation and the reduction carry over, and \cref{apx:welfare budget,apx:convex_combination} apply them to welfare maximization and to convex combinations of objectives.

\subsection{Framework for General Payoff Functions}
\label{apx:general framework}
Fix a payoff function $\RealizedPayoff_i$ for each agent $i$. The single-agent curves of \cref{sec:prelim,sub:single_curve} are defined for this payoff function by replacing revenue with payoff. The \emph{price-posting payoff curve} $\cumpayoff(\quant)$ is the seller's expected payoff from an agent who is offered the market clearing price $\marketclearing$ (\cref{def:mcp}) and buys her demand at this per-unit price. Its concave hull $\payoffHull$, the \emph{ironed price-posting payoff curve}, is the optimal expected payoff from offering a randomized per-unit price, which randomizes over at most two prices. The \emph{optimal payoff curve} $\payoffcurve(\quant)$ is the optimal expected payoff from a BIC and IIR mechanism that sells the item with probability $\quant$ in expectation, and an agent is \emph{$\zeta$-resemblant} for this objective if $\payoffHull$ is $\zeta$-resemblant to $\payoffcurve$ in the sense of \cref{def:zeta closeness}. For revenue, these definitions reduce to those in the main text.

The implementation and the reduction carry over without change. Conditional on the other agents' quantiles, an agent in the mechanism of \cref{alg:meta theorem} faces the market clearing price at her threshold quantile, so her expected contribution to the payoff equals her price-posting payoff curve evaluated at that threshold. Since payoffs are additive across agents, \cref{thm:same_payoff} holds with price-posting payoff curves in place of price-posting revenue curves, and so does \cref{thm:meta thm expected utility}, whose proof uses only \cref{thm:same_payoff}, the ex ante relaxation, and the definition of resemblance.

One difference from revenue maximization concerns the linear agents analog in \cref{sub:meta_theorem}. For revenue, the price-posting revenue curve $\cumpayoff$ of a linear agent pins down her value distribution as $\val(\quant)=\sfrac{\cumpayoff(\quant)}{\quant}$. For a general payoff function, there may be no type distribution for a linear agent whose price-posting payoff curve coincides with a given curve. The reduction does not require the existence of such an agent, since both the payoff of a per-unit pricing mechanism and the ex ante relaxation are defined directly in terms of the curves. The approximation factor $\gamma$ in \cref{thm:meta thm expected utility} can therefore be defined for the curves themselves, even when no linear agents with these curves exist. The existence of a linear agents analog matters, however, for \cref{thm:robust_optimality}. Suppose that the seller observes the expected demand curves and the price-posting payoff curves, but not the utility models. If the curves admit a linear agents analog, this analog is among the utility models consistent with the curves, and the per-unit pricing mechanism that maximizes the expected payoff given the curves, which is optimal for linear agents (\cref{sub:mrm}), is max-min optimal by the argument in \cref{apx:robust}.

\subsection{Welfare Maximization}
\label{apx:welfare budget}
For the objective of welfare maximization,
the VCG mechanism achieves the optimal welfare for linear agents.
Since the behaviors of risk-averse agents are the same as linear agents when facing per-unit prices (\cref{clm:risk_demand}), the VCG mechanism also maximizes the welfare for those agents.
However, for general non-linear agents,
VCG-style mechanisms may not exist and the optimal mechanisms for welfare maximization may be very complex.
For example, for agents with public budget constraints, \citet{mas-00} shows that the welfare-optimal mechanism takes the form of an all-pay auction for symmetric agents,
and \citet{FH-18} further show that this mechanism cannot be implemented as a DSIC mechanism.
However, by showing that budget-constrained agents resemble linear agents for the welfare,
there exist  per-unit pricing mechanisms, which are DSIC for agents with budget constraints,
that are approximately optimal for welfare maximization
even when agents have private budget constraints and when agents are asymmetric in ex ante.

\begin{proposition}
\label{prop:welfare budget}
An agent with a private budget
has an ironed price-posting welfare curve $\welfareHull$
that is $2$-resemblant
to her optimal welfare curve $\welcurve$
if the budget is drawn independently of the value.
\end{proposition}

The proof of \Cref{prop:welfare budget}
utilizes the price decomposition technique
from \citet{abr-06}
and extends it for welfare analysis.

Fix an arbitrary ex ante constraint $\quant$,
denote $\EX$ as the $\quant$
ex ante welfare-optimal mechanism,
and $\Wel{\EX}$ as its welfare.
We want to decompose $\EX$
into two mechanisms
$\EXS$ and $\EXL$
according to the market clearing price $\marketclearing$
and bound the welfare from those two mechanisms separately.
The decomposed mechanism may violate the incentive constraint for budgets,
and we refer to this setting as the
random-public-budget utility model.
Note that the market clearing price
is the same in both the private-budget model
and the random-public-budget utility model.
Intuitively, mechanism $\EXS$
contains per-unit prices at most the market clearing price,
while mechanism $\EXL$ contains per-unit prices at least the market clearing price.
Both mechanisms $\EXS$ and $\EXL$ satisfy the ex ante constraint~$\quant$,
and the sum of their welfare upper bounds
the original ex ante mechanism $\EX$,
i.e.,
$\Wel{\EX} \leq \Wel{\EXS} + \Wel{\EXL}$.

To construct $\EXS$ and $\EXL$ that satisfy
the properties above, we first introduce
a characterization
of all incentive compatible mechanisms
for a single agent with private-budget utility,
and her behavior in the mechanisms.

\begin{definition}
An \emph{allocation-payment function}
$\APF:[0,1] \rightarrow \R_+$
is a mapping from the allocation $\alloc$
to the payment $\price$.
\end{definition}

\begin{lemma}\label{lem:pricing function characterization}
For a single agent
with private-budget utility,
in any incentive compatible mechanism,
for all types with any fixed budget,
the mechanism
provides
a convex and non-decreasing allocation-payment function,
and subject to
this allocation-payment function,
each type will purchase as much as she wants
until the budget constraint binds,
or the unit-demand constraint binds, or
the value binds (i.e., her marginal utility
becomes zero).
\end{lemma}
\begin{proof}
\citet{mye-81} shows that any
mechanisms $(\alloc, \price)$ for a single
linear agent
is incentive compatible
(the agent does not prefer to misreport her value)
if and only if
a) $\alloc(\val)$ is non-decreasing;
b) $\price(\val) = \val\alloc(\val)
-
\int_0^\val \alloc(t)dt
$.
Thus,
given any non-decreasing allocation $\alloc$,
the payment $\price$
is uniquely pinned down by the incentive constraints.

Comparing with the linear utility,
the incentive compatibility
in the private-budget utility
guarantees that the agent does not prefer
to misreport
either her value or budget.
If we relax the incentive constraints
such that she is only allowed to misreport her value,
Myerson's result
already shows that
for any fixed  budget level $\wealth$,
the allocation $\alloc(\val, \wealth)$ is
non-decreasing in $\val$ and
the payment
$\price(\val,\wealth)  = \val\alloc(\val, \wealth)
-
\int_0^\val \alloc(t, \wealth)dt$
is uniquely pinned down.
We define the allocation-payment function
$\APF_\wealth(\hat\alloc) = \max\{
\price(\val, \wealth)
+
\val\cdot(\hat\alloc - \alloc(\val,\wealth)):
\alloc(\val,\wealth) \leq \hat\alloc
\}
$
if $\hat\alloc \leq \alloc(\hval, \wealth)$;
and $\infty$ otherwise.
Given the characterization of
allocation and payment above,
this allocation-payment function
is well-defined, non-decreasing and convex.
\end{proof}

This characterization is shown by relaxing the agent's incentives for misreporting the budget.
Unlike Myerson's result which gives
a sufficient and necessary condition
for incentive compatible mechanisms
for linear agents,
\Cref{lem:pricing function characterization}
only characterizes
a necessary condition for private-budget
utility.\footnote{This characterization is only necessary because it relaxes the incentive constraints for misreporting the private budget.}
This condition is already enough for our arguments.

Now we give the construction of
$\EXS$ and $\EXL$ by constructing
their allocation-payment functions.
The decomposition is illustrated in \Cref{f:decompose}.
For an agent with budget $\budget$,
let $\APF_\budget$ be the allocation-payment function
in mechanism $\EX$, and
$\alloc^*_\budget$ be the utility maximization allocation for a linear agent
with value equal to the market clearing price $\marketclearing$, i.e.,
$\alloc^*_\budget =\argmax\{\alloc : \APF_\budget'(\alloc )\leq \marketclearing\} $.
For agents with budget~$\budget$, we define the allocation-payment functions
$\APFS_\budget$ and
$\APFL_\budget$ for
$\EXS$ and $\EXL$ respectively below,
\begin{align*}
    \APFS_\budget(\alloc) &=
    \left\{
    \begin{array}{ll}
      \APF_\budget(\alloc)   &
    \text{if }\alloc \leq \alloc^*_\budget,\\
        \infty   &  \text{otherwise};
    \end{array}
    \right.
    \APFL_\budget(\alloc) =
    \left\{
    \begin{array}{ll}
      \APF_\budget(\alloc^*_\budget+\alloc) -
      \APF_\budget(\alloc^*_\budget)&
    \text{if }\alloc \leq 1 - \alloc^*_\budget,\\
        \infty   &  \text{otherwise}.
    \end{array}
    \right.
\end{align*}
By construction, for each type of the agent,
the allocation from $\EX$ is upper bounded
by the sum of the allocation from $\EXS$ and $\EXL$,
which implies that
the welfare from $\EX$ is upper bounded
by the sum of the welfare from $\EXS$ and $\EXL$,
and the requirements for the decomposition are satisfied.

\begin{figure}[t]
\begin{flushleft}
\hspace{-5pt}
\begin{minipage}[t]{0.42\textwidth}
\centering
\begin{tikzpicture}[scale = 0.4]

\draw (-0.2,0) -- (11.5, 0);
\draw (0, -0.2) -- (0, 7.2);

\begin{scope}[ultra  thick]

\draw [gray, dashed] plot [smooth, tension=0.7] coordinates { (0,0) (3, 0.8) (6, 2)};
\draw [gray, dashed] plot [smooth, tension=0.7] coordinates {(6, 7) (11, 7)};

\end{scope}

\begin{scope}%

\draw  plot [smooth, tension=0.7] coordinates { (0,0) (3, 0.8) (6, 2)};
\draw  plot [smooth, tension=0.8] coordinates {(6, 2) (9, 4) (11,6)};
\end{scope}

\draw (0, -0.8) node {$0$};
\draw (6, -0.2) -- (6, 0.2);
\draw (6, -0.8) node {$\alloc^*_\wealth$};

\draw (11, -0.2) -- (11, 0.2);
\draw (11, -0.8) node {$1$};

\draw (-0.2, 7) -- (0.2, 7);
\draw (-0.7, 7) node {$\infty$};

\draw [gray, dotted] (6, 0) -- (6, 7);
\draw [gray, dotted] (0, 7) -- (6, 7);

\draw (11, 5) node {$\APF_{\wealth}$};
\draw (6.7, 6.3)
node {$\APFS_{\wealth}$};

\end{tikzpicture}

\end{minipage}
\begin{minipage}[t]{0.57\textwidth}
\centering
\begin{tikzpicture}[scale = 0.4]

\draw (-0.2,0) -- (11.5, 0);
\draw (0, -0.2) -- (0, 6.2);

\draw (5.7,2) -- (16.2, 2);
\draw (6, 1.7) -- (6, 8.2);

\begin{scope}[ultra  thick]

\draw [gray, dashed] plot [smooth, tension=0.8] coordinates {(6, 2) (9, 4) (11,6)};
\draw [gray, dashed] plot [smooth, tension=0.8] coordinates {(11,8) (16, 8)};

\end{scope}

\draw  plot [smooth, tension=0.7] coordinates { (0,0) (3, 0.8) (6, 2)};
\draw  plot [smooth, tension=0.8] coordinates {(6, 2) (9, 4) (11,6)};

\draw [gray, dotted] (11, 0) -- (11, 8);
\draw [gray, dotted] (6, 0) -- (6, 2);
\draw [gray, dotted] (6, 8) -- (11, 8);

\draw (0, -0.8) node {$0$};

\draw (6, -0.2) -- (6, 0.2);
\draw (6, -0.8) node {$\alloc^*_\wealth$};

\draw (11, -0.2) -- (11, 0.2);
\draw (11, -0.8) node {$1$};

\draw (5.8, 8) -- (6.2, 8);
\draw (5.3, 8) node {$\infty$};

\draw (16, 1.8) -- (16, 2.2);
\draw (16, 1.2) node {$1$};

\draw
(9, 5) node {$\APFL_{\wealth}$};
\draw (4.3, 2) node {$\APF_{\wealth}$};

\end{tikzpicture}

\end{minipage}
\end{flushleft}
\caption{\label{f:decompose}
Depicted is the decomposition of the allocation-payment function.
The black lines in both figures are the allocation-payment function
$\APF_\wealth$ in ex ante optimal mechanism $\EX$;
the gray dashed lines are the allocation-payment function
$\APFS_\wealth$
and
$\APFL_\wealth$
in $\EXS$ and $\EXL$, respectively.
}
\end{figure}

As sketched above, we separately bound
the welfare in $\EXS$ and $\EXL$
by the welfare from posting
the market clearing price.
Intuitively, in $\EXS$, we are able to show that for any value higher than $\marketclearing$,
the expected allocation of the agent in $\EXS$ is weakly lower compared to market clearing
in order for the ex ante feasibility constraint to be satisfied.
Thus allocation given the market clearing price ``first-order stochastically dominates" the allocation in $\EXS$,
and hence expected welfare in $\EXS$ is lower compared to posting market clearing price~$\marketclearing$.

In $\EXL$, types with value lower than $\marketclearing$ will purchase nothing,
and types with value higher than $\marketclearing$ will purchase less compared to market clearing price as the per-unit price in $\EXL$ is higher.
Thus the expected welfare in $\EXL$ is also lower compared to posting market clearing price~$\marketclearing$.

\begin{lemma}
\label{lem:welfare ex1}
	For a single agent with random-public-budget utility,
	independently distributed value and budget,
	and any ex ante constraint $\quant$,
	the welfare from posting the market clearing price $\marketclearing$
	is at least the welfare from $\EXS$,
	i.e., $\cumwelfare(\quant) \geq \Wel{\EXS}$.
\end{lemma}
\begin{proof}
Consider an agent with type $(\val, \budget)$
and agent with type $(\val', \budget)$,
where both value $\val$ and $\val'$ are
higher than the market clearing price $\marketclearing$.
Notice that the allocations for these two types
are the same in $\EXS$ and in
market clearing,
since the per-unit price in both mechanisms
is at most $\marketclearing$ which makes
the mechanisms unable
to distinguish these two types.

Let $\alloc\primed$ be the allocation rule in $\EXS$
and let $\alloc^\quant$ be the allocation rule in posting the market clearing price $\marketclearing$.
For any value $\val \geq \marketclearing$,
the expected allocation for types
with value~$\val$
is lower in $\EXS$
than in market clearing,
i.e.,
$\expect[\budget]{\alloc\primed(\val, \budget)} \leq
\expect[\budget]{\alloc^\quant(\val, \budget)}$.
Otherwise suppose the types with value $\val^*$
has strictly higher allocation in $\EXS$ for some value $\val^* \geq \marketclearing$,
i.e, $\expect[\budget]{\alloc\primed(\val^*, \budget)}
> \expect[\budget]{\alloc^\quant(\val^*, \budget)}$.
By the fact stated in previous paragraph,
we have that
for any budget $\budget$ and any value $\val, \val^* \geq \marketclearing$,
$\alloc^\quant(\val, \budget)
= \alloc^\quant(\val^*, \budget)$,
$\alloc\primed(\val, \budget)
= \alloc\primed(\val^*, \budget)$,
and the expected allocation in $\EXS$ is
\begin{align*}
\expect[\val, \budget]{\alloc\primed(\val, \budget)}
\geq\,& \Pr[\val \geq \marketclearing]
\cdot
\expect[\val, \budget]{\alloc\primed(\val, \budget)
\given \val \geq \marketclearing} \\
=\,& \Pr[\val \geq \marketclearing]
\cdot
\expect[\budget]{\alloc\primed(\val^*, \budget)} \\
>\,& \Pr[\val \geq \marketclearing]
\cdot
\expect[\budget]{\alloc^\quant(\val^*, \budget)} \\
=\,& \Pr[\val \geq \marketclearing]
\cdot
\expect[\val, \budget]{\alloc^\quant(\val, \budget)
\given \val \geq \marketclearing}
= \quant,
\end{align*}
where the equalities hold due to
the independence between the value and the budget.
Note that this implies that $\EXS$ violates the ex ante constraint $\quant$, a contradiction.
Types with value below $\marketclearing$ may also receive a positive allocation in $\EXS$, whereas they receive nothing in market clearing; however, the total allocation of $\EXS$ is at most $\quant$, which is the total allocation in market clearing.
Further,
for any type
with value $\val \geq \marketclearing$,
$\expect[\budget]{\alloc\primed(\val, \budget)} \leq
\expect[\budget]{\alloc^\quant(\val, \budget)}$
implies that the allocation
in market clearing
``first-order stochastically dominates''
the allocation in $\EXS$, i.e.,
for any threshold $\val^\dagger$,
the expected allocation from all types with value $\val \geq \val^\dagger$
in market clearing
is at least the expected allocation from those types in $\EXS$.
Taking expectation over the value and the budget,
the expected welfare from market clearing
is at least the welfare from $\EXS$,
i.e., $\cumwelfare(\quant) \geq \Wel{\EXS}$.
\end{proof}

\begin{lemma}
\label{lem:welfare ex2}
	For a single agent with random-public-budget utility,
	independently distributed value and budget,
	and any ex ante constraint $\quant$;
	the welfare from market clearing
	is at least the welfare from $\EXL$,
	i.e., $\cumwelfare(\quant) \geq \Wel{\EXL}$.
\end{lemma}
\begin{proof}
In both $\EXL$ and market clearing,
types with value lower than $\marketclearing$ will purchase nothing,
so we only consider the types with value at least $\marketclearing$ in this proof.
Consider any type $(\val, \budget)$ where $\val \geq \marketclearing$, its allocation
in market clearing
is at least its allocation in $\EXL$,
because
the per-unit price in $\EXL$ is higher.
Thus, the welfare from market clearing
is at least the welfare from $\EXL$,
i.e., $\cumwelfare(\quant) \geq \Wel{\EXL}$.
\end{proof}

\begin{proof}[Proof of \Cref{prop:welfare budget}]
Combining \Cref{lem:welfare ex1,lem:welfare ex2},
for any quantile $\quant$, we have
\begin{equation*}
\welcurve(\quant) = \Wel{\EX}
\leq \Wel{\EXS} + \Wel{\EXL}
\leq 2\cumwelfare(\quant)
\leq \max_{\quant' \leq \quant}
2\welfareHull(\quant').
\qedhere
\end{equation*}
\end{proof}

\subsection{Convex Combinations of Objectives}
\label{apx:convex_combination}
For more complex payoff functions such as the convex combination of welfare and revenue maximization,
we can extend our previous results
to show that if an agent resembles linear agents for both welfare maximization and revenue maximization,
this agent resembles linear agents for any convex combination of the two objectives.
\begin{lemma}
\label{lem:resemble_convex}
If an agent is $\zeta$-resemblant for objective $1$
and $\zeta'$-resemblant for objective~$2$ with non-negative values,
then this agent is $(\zeta+\zeta')$-resemblant for any convex combination of the two objectives.
\end{lemma}
\begin{proof}
For any quantile $\quant$, let $\EX$ be the $\quant$ ex ante optimal mechanism for the convex combination of the objectives.
Let $\Payoff[1]{\EX}$ be the contribution of objective~$1$
given mechanism $\EX$ and
$\Payoff[2]{\EX}$ be the contribution of objective $2$
given mechanism $\EX$.
Let $\Payoff{\EX} = \alpha\cdot \Payoff[1]{\EX}+(1-\alpha)\cdot\Payoff[2]{\EX}$
be the convex combination of the contributions
given $\alpha\in (0,1)$.
Let $\quant_1 = \argmax_{\quant'\leq \quant} \concaveHull_1(\quant')$
and $\quant_2 = \argmax_{\quant'\leq \quant} \concaveHull_2(\quant')$,
where $\concaveHull_1$ and $\concaveHull_2$
are the concave hull of price-posting payoff curves for objectives $1$ and $2$ respectively.
Let $\concaveHull$ be the concave hull of price-posting payoff curves for the convex combination of objectives $1$ and $2$.
Then, we have
\begin{align*}
&\Payoff{\EX}
= \alpha\cdot \Payoff[1]{\EX}+(1-\alpha)\cdot\Payoff[2]{\EX}\\
\leq\,& \alpha\zeta\cdot \concaveHull_1(\quant_1)
+(1-\alpha)\zeta'\cdot\concaveHull_2(\quant_2)
\leq \zeta\cdot \concaveHull(\quant_1)
+\zeta'\cdot\concaveHull(\quant_2)
\leq (\zeta+\zeta')
\max_{\quant'\leq\quant}\concaveHull(\quant').
\end{align*}
Thus this agent is $(\zeta+\zeta')$-resemblant for the convex combination of the two objectives.
\end{proof}

\section{Public Budget}
\label{apx:rev close public budget}

We consider the setting in which agents have deterministic public budgets. This situation can be seen as a special case of private-budget agents, where we consider a point-mass budget distribution. As a sanity check, the independence assumption holds for public-budget agents.

\begin{restatable}{proposition}{propRegularPublicBudget}
\label{thm:public budget regular}
A public-budget agent is $1$-resemblant if her marginal value distribution is regular.
\end{restatable}
The proof of \Cref{thm:public budget regular} relies on the technique of Lagrangian relaxation on the budget constraint developed in \citet{AFHH-13} and \citet{FH-18}. Combined with our reduction framework in \cref{thm:meta thm expected utility}, the interpretation of \Cref{thm:public budget regular} is that for public-budget agents with regular marginal value distributions, the worst-case approximation ratio of per-unit pricing mechanisms to the ex ante relaxation occurs when the public budgets do not bind.
Moreover, in \cref{thm:public budget irregular}, we also show that the regularity assumption is not essential for approximation. That is, without the assumption of a regular value distribution, a public-budget agent is still $2$-resemblant.

\begin{proof}[Proof of \cref{thm:public budget regular}]
For an agent with public budget $\wealth$,
the $\exquant$ ex ante optimal mechanism
is the solution of the following program,
\begin{align}
\label{eq:prog 1}
\begin{array}{cl}
\max\limits_{(\alloc,\price)}
\,\, &\expect[\val]{\price(\val)}
\\
\text{s.t.}\,\, &(\alloc,\price) \text{ are IC, IR},
\\
&\expect[\val]{\alloc(\val)} = \exquant,
\\
&\price(\hval) \leq \wealth.
\end{array}
\end{align}
where $\hval$ is the highest possible
value of the agent.
Consider the Lagrangian
relaxation of the budget constraint
in \eqref{eq:prog 1},
\begin{align}
\label{eq:prog 2}
\begin{array}{cl}
\min\limits_{\lagrange \geq 0}
\max\limits_{(\alloc,\price)}
\quad &\expect[\val]{\price(\val)}
+
\lagrange \wealth - \lagrange \price(\hval) \\
\text{s.t.} &(\alloc,\price) \text{ are IC, IR}, \\
&\expect[\val]{\alloc(\val)} = \exquant.
\end{array}
\end{align}
Let $\lagrange^*$ be
the optimal solution in
program~\eqref{eq:prog 2}.
If we fix $\lagrange = \lagrange^*$ in
program \eqref{eq:prog 2},
its inner maximization program
can be thought of as a $\exquant$
ex ante optimal mechanism design
for a linear agent with
Lagrangian objective function
$\expect[\val]{\price(v)} -
\lagrange^* \price(\hval)$.
Thus, we
define the Lagrangian price-posting
revenue curve $\lcumprice(\cdot)$
where $\lcumprice(\quant)$
is the maximum
value of
the Lagrangian objective $\expect[\val]{\price(\val)}
-\lagrange^* \price(\hval)
$
in price-posting mechanism with
per-unit price $\val(\quant)$.
For any
$\quant \in (0, 1]$, by the definition,
$\lcumprice(\quant)
= \quant \val(\quant) - \lagrange^* \val(\quant)$.
For $\quant = 0$, notice that
the agent with
$\hval$ is indifferent
between purchasing or not purchasing.
Thus,
by the definition,
$\lcumprice(\quant) = 0$ if $\quant = 0$.

Now, we consider the concave hull of
the Lagrangian price-posting revenue curve
$\lcumprice(\cdot)$ which
we denote as $\clcumprice(\cdot)$.
Let $\quant\primed$ be the smallest
solution
of equation
$\lcumprice(\quant) =
\quant\lcumprice'(\quant)$ in $(0,1)$, i.e., the smallest stationary point of $\sfrac{\lcumprice(\quant)}{\quant}$, and let $\quant\primed=1$ if no such solution exists.
In the latter case, the concave hull is the chord from the origin to $(1,\lcumprice(1))$, and the argument below applies with $\lcumprice'(\quant\primed)$ replaced by $\lcumprice(1)$.
Then, for any $\quant \leq \quant\primed$,
$\clcumprice(\quant) =
\quant\,\lcumprice'(\quant\primed)$.
For any $\quant \geq \quant\primed$,
we show $\clcumprice(\quant) = \lcumprice(\quant)$
by the following arguments.
First notice that
$\lcumprice(\quant\primed) \geq 0$,
and hence $\quant\primed \geq \lambda^*$.
Consider $\lcumprice''(\quant) = \val''(\quant)(\quant - \lambda^*) + 2\val'(\quant)$.
Clearly, $\val'(\quant) \leq 0$.
If $\val''(\quant) \leq 0$, then
$\lcumprice''(\quant) \leq 0$.
If $\val''(\quant) > 0$,
then
$\lcumprice''(\quant) = \val''(\quant)(\quant - \lambda^*) + 2\val'(\quant) \leq \quant
\val''(\quant) + 2\val'(\quant) \leq 0$,
where $\quant \val''(\quant) + 2\val'(\quant)$ is
non-positive due to the regularity of the
value distribution.

To summarize, $\clcumprice(\cdot)$,
the concave hull of the Lagrangian price-posting revenue
curve
satisfies
\begin{align*}
    \clcumprice(\quant) =
    \left\{
    \begin{array}{ll}
    \quant\,\lcumprice'(\quant\primed)
    & \text{ if } \quant \in
    [0, \quant\primed] \\
    \lcumprice(\quant) & \text{ if } \quant \in
    [\quant\primed, 1]
    \end{array}
    \right.
\end{align*}
Therefore, using a similar ironing
technique based on the revenue curves
for linear agents with irregular value
distribution \citep[e.g.][]{mye-81, BR-89, AFHH-13},
\Cref{lem:lagrangian relaxation} (stated below)
implies that
the $\exquant$ ex ante optimal mechanism
for the Lagrangian objective
irons quantiles between $[0, \quant\primed]$
under $\exquant$ ex ante constraint,
which is still a posted-pricing mechanism; by the Lagrangian duality argument of \citet{AFHH-13} and \citet{FH-18}, this mechanism is also optimal in program~\eqref{eq:prog 1}.
\end{proof}

\begin{lemma}[\citealp{AFHH-13}]\label{lem:lagrangian relaxation}
For any incentive compatible and
individually rational mechanism
$(\alloc(\cdot), \price(\cdot))$
and an agent with any
Lagrangian price-posting revenue curve $\lcumprice(\quant)$,
the expected Lagrangian objective of the agent
is upper-bounded by her expected marginal
Lagrangian objective
of the same allocation rule, i.e.,
\begin{align*}
\expect[\val]{\price(\val)} - \lagrange^*\price(\hval)
\leq \expect[\quant]{\clcumprice'(\quant)\cdot \alloc(\val(\quant))}.
\end{align*}
Furthermore, this inequality holds with equality if the allocation
rule $\alloc(\cdot)$ is constant on all intervals of values $\val(\quant)$
where $\clcumprice(\quant) > \lcumprice(\quant)$.
\end{lemma}

For an agent with a
general value distribution,
resemblance follows from
the characterization that the ex ante optimal mechanism has menu size at most 2 \citep{AFHH-13}.
Thus offering the optimal single menu pricing mechanism, i.e., posting the market clearing price, is a 2-approximation.
\begin{lemma}[\citealp{AFHH-13}]
\label{lem:public budget irregular}
 For a single agent with public budget,
  the $q \in [0, 1]$ ex ante optimal mechanism
  has a menu with size at most two.
\end{lemma}

\begin{proposition}
\label{thm:public budget irregular}
An agent with public budget
has the ironed price-posting revenue curve $\concaveHull$
that is $2$-resemblant
to her optimal revenue curve $\revcurve$.
\end{proposition}
\begin{proof}
By \Cref{lem:public budget irregular},
the allocation rule $\allocq$ of the ex ante revenue maximization mechanism
for the single agent with public budget
has a menu of size at most two.
We decompose its allocation into
$\alloclow$ and $\allochigh$
as illustrated in \Cref{f:alloc decompose}.
Note that both allocations $\alloclow$
and $\allochigh$ are (randomized)
price-posting allocation rules,
and neither allocation violates the allocation constraint $\quant$.
Thus,
\begin{equation*}
\revcurve(\quant) = \Rev{\allocq}
= \Rev{\alloclow} + \Rev{\allochigh}
\leq 2\,
\max_{\quant\primed \leq \quant}
\concaveHull(\quant\primed).
\qedhere
\end{equation*}
\end{proof}

\begin{figure}[t]
\begin{flushleft}
\hspace{-10pt}
\begin{minipage}[t]{0.48\textwidth}
\centering
\begin{tikzpicture}[scale = 0.5]

\draw (-0.2,0) -- (11.5, 0);
\draw (0, -0.2) -- (0, 5);

\draw (3, 0) -- (3, 2);
\draw (3, 2) -- (6, 2);
\draw (6, 2) -- (6, 4);
\draw (6, 4) -- (11, 4);

\begin{scope}[ultra  thick]

\draw [dashed] (0, 0) -- (3, 0);
\draw [dashed] (3, 0) -- (3, 2);
\draw [dashed] (3, 2) -- (11, 2);

\end{scope}

\draw (0, -0.8) node {$0$};
\draw (-0.2, 4.8) -- (0.2, 4.8);
\draw (-0.6, 4.8) node {$1$};

\draw (11, -0.2) -- (11, 0.2);
\draw (11, -0.8) node {$\val$};

\draw (10.7, 4.5) node {$\allocq$};
\draw (10.7, 2.5) node {$\alloclow$};

\end{tikzpicture}

\end{minipage}
\begin{minipage}[t]{0.48\textwidth}
\centering
\begin{tikzpicture}[scale = 0.5]

\draw (-0.2,0) -- (11.5, 0);
\draw (0, -0.2) -- (0, 7);

\draw (-0.2,2) -- (11.5, 2);

\draw (3, 0) -- (3, 2);
\draw (3, 2) -- (6, 2);
\draw (6, 2) -- (6, 4);
\draw (6, 4) -- (11, 4);

\begin{scope}[ultra  thick]

\draw [dashed] (0, 2) -- (6, 2);
\draw [dashed] (6, 2) -- (6, 4);
\draw [dashed] (6, 4) -- (11, 4);

\end{scope}

\draw (-0.6, 2) node {$0$};
\draw (-0.2, 6.8) -- (0.2, 6.8);
\draw (-0.6, 6.8) node {$1$};

\draw (11, 1.8) -- (11, 2.2);
\draw (11, 1.2) node {$\val$};

\draw (10.7, 4.5) node {$\allocq$};
\draw (10.7, 3.5) node {$\allochigh$};

\end{tikzpicture}

\end{minipage}
\end{flushleft}
\caption{\label{f:alloc decompose}
The thin solid line is the allocation rule
for the optimal ex ante mechanism.
The thick dashed line on the left side
is the allocation of the decomposed mechanism with lower price,
while the thick dashed line on the right side
is the allocation of the decomposed mechanism with higher price.
}
\end{figure}

\section{Numerical Illustrations}
\label{apx:example_detail}
In this section, we provide the details for numerical evaluations. Recall that we will focus on selling $m$ units of identical items to $n\geq m$ i.i.d.\ agents with the objective of revenue maximization.

\subsection{Risk-Averse Utility}
Recall that we focus on a simple example
where each agent has a constant absolute risk aversion (CARA) utility function,
i.e.,
\begin{align*}
\risk(z) = \frac{1}{a}\rbr{1-\exp\rbr{-az}}
\end{align*}
for risk parameter $a>0$, where $\exp(\cdot)$ is the exponential function.
We assume that the risk-averse agent has private value drawn from the uniform distribution on $[0,1]$.

By \cref{clm:risk_demand}, a risk-averse agent behaves the same as a linear agent when facing a per-unit price and hence the risk attitude of the agent does not matter under posted pricing.
Therefore, we can apply the same sequential posted pricing mechanism
for risk-averse agents
as the one for linear agents,
and achieve the same expected revenue as if all agents have linear utilities.
In particular, we will numerically evaluate the revenue of sequential posted pricing mechanism where the per-unit price $\pricei$ of each agent is set such that the probability the price is accepted equals $\min\lbr{\frac{m}{n}, \frac{1}{2}}$,
where~$\frac{1}{2}$ is the probability the item is sold in the single-agent problem when the seller posts a monopoly price of $\frac{1}{2}$ to maximize the expected revenue for uniform distributions.
Thus, the per-unit price $\pricei$ equals $1-\min\lbr{\frac{m}{n}, \frac{1}{2}}$ for each agent~$i$.
The revenue guarantees of this particular sequential posted pricing are illustrated in \Cref{fig:rev_risk}.
The upper bound in \Cref{fig:rev_risk} is the ex ante relaxation with the optimal revenue curve replaced by the expected allocated value, $\payoffcurve(\quant)\leq\int_0^{\quant}\val(z)\,\dd z$ (\cref{lem:risk averse revenue upper bounded by allocated value}). Neither this bound nor the revenue of sequential posted pricing depends on the risk parameter $a$, so the comparison holds for every degree of risk aversion.

\subsection{Private-Budget Utility}
Recall that we consider budget-constrained agents whose value and budget are drawn independently from the uniform distribution on $[0,1]$.
We numerically evaluate the revenue of sequential posted pricing mechanism where the per-unit price of each agent is set such that the probability each agent wins a unit of item is $\min\{\frac{m}{n}, \frac{1+\sqrt{3}}{6}\}$,
where $\frac{1+\sqrt{3}}{6}$ is the probability an item is sold in the single-agent problem when the seller posts a per-unit price to maximize the expected revenue.
Therefore, the per-unit price for each budget-constrained agent $i$ in this example is
\begin{align*}
\pricei = \frac{1}{2}\rbr{3-\sqrt{9-8(1-\min\{\tfrac{m}{n}, \tfrac{1+\sqrt{3}}{6}\})}}.
\end{align*}
The revenue guarantees of this particular sequential posted pricing are illustrated in \Cref{fig:rev_budget}.
The upper bound in \Cref{fig:rev_budget} is the ex ante relaxation with the optimal revenue curve of a single budget-constrained agent computed numerically by a linear program over a discretized type space.

\section{Tie-Breaking Rules}
\label{apx:tie}
In \cref{alg:meta theorem}, it is possible that based on the tie-breaking rule of the per-unit pricing mechanism,
agent $i$ may not receive the item if her quantile equals the threshold quantile.
In this case, the probability of agent~$i$ with type $\typei$
for winning an item when others have quantiles $\quant_{-i}$
can be strictly smaller than $\demand^{\typei}(\marketclearing[\exquanti])$.
Denote this probability as $z_i(\typei, \quant_{-i})$.
In this case, the payment of agent $i$ with type $\typei$
in mechanism $\mech$ is $\marketclearing[\exquanti]\cdot z_i(\typei, \quant_{-i})$ regardless of the allocation.

Note that based on our construction, the marginal quantile for each agent is drawn from the uniform distribution on $[0,1]$.
Therefore, ties occur with measure zero and hence the expected revenue of mechanism $\mech$ for non-linear agents is not affected.

Now it is sufficient to show that both DSIC and IIR properties in mechanism $\mech$ for non-linear agents hold
even for the measure zero event that
\begin{align*}
z_i(\typei, \quant_{-i})<\demand^{\typei}(\marketclearing[\exquanti]).
\end{align*}
Note that for each agent $i$ with type $\typei$,
her quantile is drawn from cumulative distribution function $H_i^{\typei}$.
Therefore, such an event happens only if distribution $H_i^{\typei}$ has a point mass at threshold quantile $\exquanti$.
Moreover,
\begin{align*}
z_i(\typei, \quant_{-i}) = \lim_{q\to\hat{\quant}^{-}_{i}} \demand^{\typei}(\marketclearing).
\end{align*}
Since the demand correspondence is upper hemi-continuous (\cref{asp:ordinary good}),
lottery $z_i(\typei, \quant_{-i})$ is also a best response of agent $i$ with type $\typei$ given market clearing price $\marketclearing[\exquanti]$.
Therefore, the constructed mechanism is also DSIC and IIR for the measure zero event.

\end{document}